\documentclass[10pt,journal,compsoc]{IEEEtran}
\pdfoutput=1
\usepackage{algorithm}
\usepackage{algpseudocode}

\usepackage[official]{eurosym}
\usepackage{graphicx}
\usepackage{subfig}
\usepackage{float}
\usepackage{booktabs}
\usepackage{tabularx}
\usepackage[keeplastbox]{flushend}
\usepackage{todonotes}
\usepackage{lipsum}
\usepackage{amsfonts}
\usepackage{mathtools}
\usepackage{amssymb}
\usepackage{amsthm} % for theorems
\usepackage{amsmath}
\usepackage{enumitem}
\usepackage{multicol}
\usepackage{threeparttable}
\usepackage{tikz}
\usepackage{bm}
\usepackage{tikzscale}
\usepackage{dsfont}
\usepackage{multirow}
\usepackage{array}
\usepackage{makecell}
\usepackage{adjustbox}
\usepackage{pgflibraryshapes}
\usepackage[makeroom]{cancel}
\usepackage{cuted}
\usepackage{relsize}
\usepackage[normalem]{ulem}

\newtheorem{lemma}{Lemma}
\newtheorem{theorem}{Theorem}

\usetikzlibrary{calc}
\usetikzlibrary{shapes,snakes}
\usetikzlibrary{patterns}

\usepackage{soul}

\usepackage{mdframed}

\makeatletter
\newenvironment{problemformulation}[1][htb]%
  {\renewcommand{\ALG@name}{Formulation}% Update algorithm name
   \begin{algorithm}[#1]%
  }{\end{algorithm}}
\makeatother

\begin{document}

%[Delay and reliability-constrained VNF placement on mobile and volatile 5G infrastructure]
\title{Delay and reliability-constrained VNF placement on mobile and volatile 5G infrastructure}

%% % Balazs
%% \author{\IEEEauthorblockN{Bal\'azs N\'emeth},
%% \IEEEauthorblockA{
%% MTA-BME Network Softwarization Research Group,
%% balazs.nemeth@tmit.bme.hu}\\
%% \and
%% % Nuria
%% \IEEEauthorblockN{Nuria Molner},
%% \IEEEauthorblockA{IMDEA Networks Institute, Spain;
%% Universidad Carlos III de Madrid, Spain;
%% nuria.molner@imdea.org}\\
%% \and
%% % Jorge
%% \IEEEauthorblockN{Jorge Mart\'in-P\'erez}
%% \IEEEauthorblockA{Universidad Carlos III de Madrid, Spain;
%% jmartinp@it.uc3m.es}\\
%% \and
%% % Carlos
%% \IEEEauthorblockN{Carlos J. Bernardos}
%% \IEEEauthorblockA{Universidad Carlos III de Madrid, Spain;
%% cjbc@it.uc3m.es}\\
%% \and
%% % Antonio
%% \IEEEauthorblockN{Antonio de la Oliva}
%% \IEEEauthorblockA{Universidad Carlos III de Madrid, Spain;
%% aoliva@it.uc3m.es}\\
%% \and
%% % Balazs Sonkoly
%% \IEEEauthorblockN{Bal\'azs Sonkoly}
%% \IEEEauthorblockA{MTA-BME Network Softwarization Research Group,
%% sonkoly@tmit.bme.hu}
%% 
%% \IEEEauthorrefmark{1}MTA-BME Network Softwarization Research Group,
%% 
%% }

\author{Bal\'azs N\'emeth, Nuria Molner, Jorge Mart\'in-P\'erez,\\ Carlos J. Bernardos, Antonio de la Oliva, and Bal\'azs Sonkoly%
\IEEEcompsocitemizethanks{
\IEEEcompsocthanksitem Bal\'azs N\'emeth and Bal\'azs Sonkoly are with MTA-BME Network Softwarization Research Group, Hungary%
\IEEEcompsocthanksitem Emails: balazs.nemeth@tmit.bme.hu, sonkoly@tmit.bme.hu
\IEEEcompsocthanksitem Nuria Molner is with IMDEA Networks Institute, Spain
\IEEEcompsocthanksitem Email: nuria.molner@imdea.org
\IEEEcompsocthanksitem Nuria Molner, Jorge Mart\'in-P\'erez, Carlos J. Bernardos, and Antonio de la Oliva  are with Universidad Carlos III de Madrid, Spain
\IEEEcompsocthanksitem Emails: ----, jmartinp@it.uc3m.es, cjbc@it.uc3m.es, aoliva@it.uc3m.es}%
}

% \author{Michael ̃Shell, ̃\IEEEmembership{Member, ̃IEEE,} John ̃Doe, ̃\IEEEmembership{Fellow, ̃OSA,} and ̃Jane ̃Doe, ̃\IEEEmembership{Life ̃Fellow, ̃IEEE}%
% \IEEEcompsocitemizethanks{\IEEEcompsocthanksitem M.Shell is with the Georgia Institute of Technology.\IEEEcompsocthanksitem J. Doe and J. Doe are with Anonymous University.}%
% \thanks{Manuscript received January 20, 2002; revised August 26, 2015.}}
\maketitle

\begin{abstract}
    The ongoing research and industrial exploitation of SDN and NFV technologies promise higher flexibility on network automation and infrastructure optimization.  Choosing the location of Virtual Network Functions is a central problem in the automation and optimization of the software-defined, virtualization-based next generation of networks such as 5G and beyond.  Network services provided for autonomous vehicles, factory automation, e-health and cloud robotics often require strict delay bounds and reliability constraints, which are influenced by the location of its composing Virtual Network Functions. Robots, vehicles and other end-devices provide significant capabilities such as actuators, sensors and local computation which are essential for some services. Moreover, these devices are continuously on the move and might lose network connection or run out of battery power, which further challenge service delivery in this dynamic environment.
    This work tackles the mobility, and battery
    restrictions; as well as the temporal aspects and conflicting traits of reliable, low latency service deployment over a volatile network, where mobile compute nodes act as an extension of the cloud and edge computing infrastructure. The problem is formulated as a cost-minimizing Virtual Network Function placement optimization and an efficient heuristic is proposed. The algorithms are extensively evaluated from various aspects by simulation on detailed real-world scenarios.
\end{abstract}

\section{Introduction}
\label{sec:introduction}

\begin{figure*}[t!]
    \begin{tikzpicture}
         \node [anchor=south west] (label) at (0,0) {\includegraphics[width=\textwidth]{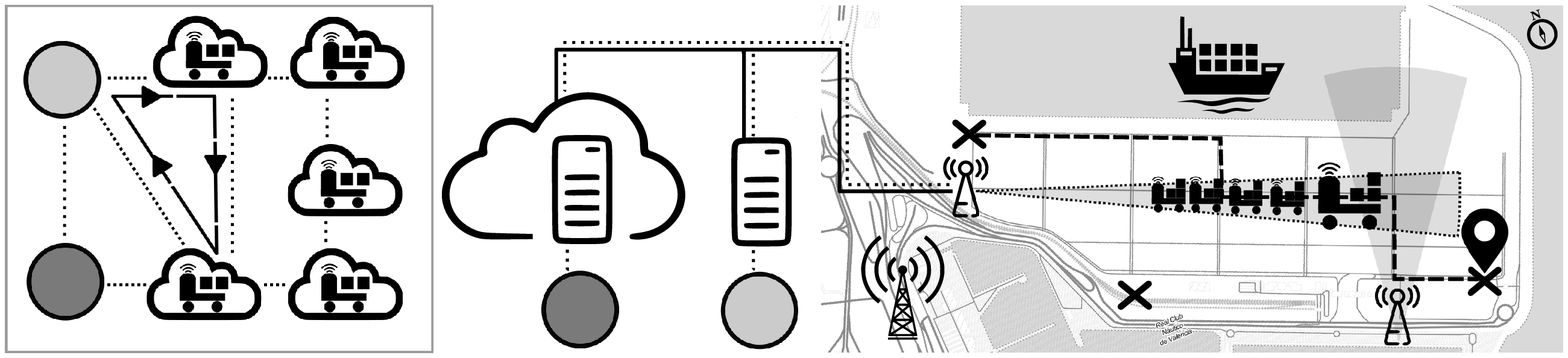}};
         \node[anchor=west] at (1.8,-0.2) {NS graph};
         \node[anchor=west] at (1.75,2.6) {\small SFC};
         \node[align=center,anchor=west] at (7.7,-0.3) {Remote \\ control VNF};
         \node[align=center,anchor=west,rotate=90] at (8.27,1.12) {\footnotesize Edge server};
         \node[align=center,anchor=west] at (6.2,-0.3) {DB \\ VNF};
         \node[align=center,anchor=west,rotate=90] at (6.22,0.97) {\footnotesize Cloud server};

         % Antennas
         \node[fill=white,rounded corners,align=center,anchor=west,thin,draw=black] at (9.9,0) {LTE};
         \node[fill=white,rounded corners,align=center,anchor=west,thin,draw=black] at (15.65,-0.1) {NR};

         % Locations
         \node[fill=white,align=center,anchor=west,thin,draw=black] at (12.7,0.3) {D2};
         \node[fill=white,align=center,anchor=west,thin,draw=black] at (16.7,0.5) {D1};
         \node[fill=white,align=center,anchor=west,thin,draw=black] at (10.9,3.1) {S};
     \end{tikzpicture}
     \caption{Deployment of a cloud robotics warehousing NS in Valencia city haven.}
%% === VERBOSE CAPTION ===
%     \caption{Deployment of a cloud robotics NS for goods delivery in Valencia city haven.
%         The Figure schematically illustrates some RUs, and
%     shadowed regions represent NR RUs' Line of Sight (LoS) coverage
%     areas, due to the containers stacked in the haven.
% The NR RU with a dotted coverage area border gives connection to the
% robot cluster delivering goods.
% Robots obey the Remote control VNF and travel along the dashed line from location N towards destination W.
% \emph{Slave} robots (smaller ones) follow the \emph{master} robot,
% and all of them report context information to the Database (DB) VNF.}
    \label{fig:use-case-illustration}
\end{figure*}

% \begin{itemize}

%     \item mention how private/public networks may work for industrial scenarios of 5G;

%     \item talk about cloud robotics and strict URLLC;
%     \item mention factory/haven robots carrying objects around
%     \item mention drone swarns and (autonomous) cars as examples for moving computations nodes
%     \item mention some IoT/fog background;

%     \item talk about dive/coral/any project to which we give ACK;

% \end{itemize}

% why
5G (\& beyond) systems have been promising several appealing sometimes
unbelievable future use-cases and applications which could reshape our
society.  The vast number of IoT devices, autonomous vehicles and
different types of robots collaborating with each other and with
humans are expected to be part of our lives.  These devices usually
require coordination or fine granular, dynamically programmable
control from a reliable and permanently available platform.
Coordination of collaborating robots, drone swarms, self-driving cars
or any types of unmanned vehicles are good examples with several
fields of application from industry to agriculture and from logistics
to emergency management.  The envisioned use-cases typically pose
serious challenges on the underlying networks and cloud platforms in
terms of latency and reliability.  For example, 3GPP specified a
dedicated set of features for mission critical applications referred
to as Ultra-Reliable Low-Latency Communication (URLLC)
\cite{3gpp-urllc}.

Edge and fog computing, multi-access edge
computing (MEC) are key enablers of these applications.  The main
concept is to extend traditional cloud computing by deploying compute
resources closer to customers and end devices. By these means, both end devices and central cloud servers can offload computational tasks to resources at the edge or the fog
resulting in lower delays and in reduced network load.  In order to
meet the strict delay and reliability requirements of mission critical
applications, a distributed and heterogeneous infrastructure and the
encompassed compute and network resources should be managed carefully.
The underlying infrastructure includes both public and private
cloud/edge resources~\cite{whitepaperPrivate} providing execution environments for Virtual
Network Functions (VNFs) interconnected by public 5G networks and
privately operated domains.  In this environment, resource
orchestration is a challenging task which aims at always finding the
proper placement of software components realizing the service.
Moreover, robots or different vehicles equipped with sensors,
actuators and local computation environments, provide capabilities
which can or must be consumed by certain applications.  More exactly,
now we can run VNFs on these continuously moving mobile devices, and
the uninterrupted communication to other service components should
also be guaranteed. Beside mobility, the limited battery capacity and
the VNFs power consumption are novel aspects to be considered
in the placement decision.

% what
This paper addresses the temporal aspects and conflicting
traits of reliable, low latency service deployment over a
network where mobile compute nodes (e.g., robots, drones) act as an
extension of the cloud and edge computing infrastructure.

% how
The research contribution is threefold.
First, the VNF placement problem is formulated as a cost-minimizing
optimization problem. The present work
extends formulations in the state of the art imposing
the radio coverage of mobile fog devices, and
preventing that VNF deployments use fog devices
that may run out of battery.
Second, the optimization problem is solved by a novel
heuristic algorithm that, to the best of the authors'
knowledge, is the first one getting close to optimal
results while tackling both radio coverage, and
battery restrictions of fog environments.
Finally, the proposed algorithms are evaluated
via extensive simulations on a real-world scenario. 
The results confirm the beneficial
properties of our solutions in terms of scalability, cost, and runtime.
% Results confirm that the proposed heuristic achieves close to optimal deployment costs while tackling delay, reliability and mobility issues.

The rest of the paper is organized as follows.
Sec.~\ref{sec:use-case} introduces a future use-case motivating our work.  Sec.~\ref{sec:problem-formulation} is
devoted to the detailed description of our model and the optimization
problem.  Sec.~\ref{sec:heuristics} describes the proposed heuristic
algorithm.  Sec.~\ref{sec:results} presents the algorithms' evaluation
from different aspects based on extensive simulations.  In
Sec.~\ref{sec:relatedwork}, a summary on the related work is given
while Sec.~\ref{sec:conclusion} draws the conclusions.

%%% Local Variables: 
%%% mode: latex
%%% ispell-local-dictionary: "en_US"
%%% TeX-master: "main"
%%% End: 

\section{Use Case: mobile robotics}
\label{sec:use-case}

This work tackles the mobile robotics use case~\cite{3gppMobileRobotics} as a warehousing solution for future factories. % of the future.
In particular, it deals with the transport of goods from boats to specific locations of Valencia city haven.

The use case considers a cluster of robots, that move in a \emph{master-slave} fashion to deliver goods arriving to the haven.
Each of the robots carries containers from a pick up point (S in Figure~\ref{fig:use-case-illustration}) to a drop off point (D1 and D2).
In particular, the \emph{master} robot is followed by the other \emph{slave} robots (represented in smaller size in Figure~\ref{fig:use-case-illustration}) of the cluster along its way towards the drop off point.
Robots can communicate among themselves to report position status, or other context information useful for the \emph{master-slave} coordination.
Thus, robots have device-to-device communication between them, and computational capabilities so they can execute lightweight VNFs~\cite{lightweight} (represented as robots inside a cloud in Figure~\ref{fig:use-case-illustration}) to perform tasks such as periodic communication of context information.

To allow remote driving, the cluster of robots communicate with Radio Units (RUs). RUs proximity is needed as Mobile robotics demand communications with cycle times between 1ms and 100ms (for machine control, and video operated remote control cases) \cite{3gppMobileRobotics}.
%Since \cite{3gppMobileRobotics}~states that mobile robotics demand communications with cycle times between 1ms and 100ms (for machine control, and video operated remote control cases), RUs proximity is necessary to make possible delays demanded by the mobile robotics use case.

While moving, robots may run out of battery, or leave the coverage area of the RUs that they use to communicate with the servers holding remote driving VNFs (light dark ball in the Network Service (NS) graph of Figure~\ref{fig:use-case-illustration}).
So it is important to take into account (i) that a robot is not selected for goods delivery if it may run out of battery; and (ii) the mobility between RUs.

To increase the RUs coverage and improve the end-to-end (e2e) delay, the use case presented in this section considers that the haven 
is covered by LTE RUs managed by a network operator, and New Radio (NR) RUs belonging to its Non Public Network (NPN). This is called a NPN deployment in a public network ~\cite{whitepaperPrivate}.
%According to~\cite{whitepaperPrivate}, this means that this section use case, and Figure~\ref{fig:use-case-illustration} follow a NPN deployed in public network.
That is, Valencia city haven only owns the NR RUs, and its management (subscription, gateways, control plane) is done by the public network, i.e., a network operator.

For the public network infrastructure, a 5G transport network is assumed based on~\cite{ituArch} and~\cite{luca-5g-qos}.
All the RUs present in the use case transmit their traffic up to an access ring composed of several switches connected in a ring fashion.
The traffic of the access rings is latter gathered by the aggregation rings which forward traffic up to the core of the public infrastructure.
The presented use case, assumes that cloud servers are in the core of the public network, edge servers are co-located next to the access ring and the aggregation ring switches.
Regarding computational resources (i.e., CPU, memory and disk), edge servers in access rings are less powerful than edge servers in aggregation rings, and cloud servers are more powerful than edge servers.

\section{Problem formulation}
\label{sec:problem-formulation}
This section presents the formulation of the use case to tackle the VNF allocation as an optimization problem.
%In addition to the proposed heuristic, t
The problem is solved using an integer program solver to gain optimality and scalability insights.
%\change{}{THIS SENTENCE CAN BE REMOVED TO SHORTEN THE PAPER-->}The section is divided in two parts, the system model and the mathematical formulation of the problem. This facilitates a better understanding of the problem and its requirements.

\subsection{System model}
\label{sec:model}
%\subsection{Infrastructure and services graph}
%\label{subsec:graph}
%
\begin{table*}[htb]
\caption{Definition of variables and parameters}
\label{tab:notation}
\centering
\begin{adjustbox}{width=0.7\textwidth}
\begin{tabular}{c|c||c}
\multicolumn{2}{c||}{\textbf{Notation}} & \textbf{Definition} \\ 
\hline
\hline
\multicolumn{3}{c}{\textbf{Parameters}} \\ 
\hline
\hline
\multicolumn{2}{c||}{$G_I$} & Network infrastructure graph \\
\hline
\multirow{4}{*}{$G_I$} & $V(G_I)$ &   All infrastructure nodes  \\
\cline{2-3}
& \multirow{2}{*}{$V_{*}(G_I)$} & Nodes of type $*$ in the infrastructure \\
& & $* \in \{AP, \ S, \ M, \ H, \ RC_q \}$ \\
%& $V_{AP}(G_I)$ & Access Points in the infrastructure \\
%\cline{2-3}
%& $V_{S}(G_I)$ & Edge or cloud server nodes \\
%\cline{2-3}
%& $V_{H}(G_I)$ & Infrastructure nodes with computation \\
%\cline{2-3}
%& $V_{M}(G_I)$ & All mobile nodes of all robot cluster \\
%\cline{2-3}
%& $V_{RC_q}(G_I)$ & Nodes of robot cluster $RC_q$ \\
\cline{2-3}
& $E(G_I)$ & Edges of the infrastructure \\
\hline
\multicolumn{2}{c||}{$G_S$} & Network service graph \\
\hline
\multirow{8}{*}{$G_S$} & $V(G_S)$ & All VNFs of the network service \\
\cline{2-3}
& $\mathcal{P}(G_S)$ & All paths of the service graph \\
\cline{2-3}
& $G_s$ & Graph of SFC $G_s$ \\
\cline{2-3}
& $V(G_s)$ & VNFs of SFC $G_s$ \\
\cline{2-3}
& $E(G_s)$ & Edges of SFC path $G_s$  \\
\cline{2-3}
& $\Delta_{G_s}$ & Delay requirement of SFC $G_s$ \\
\cline{2-3}
& $th_{bat}^{s}$ & Battery threshold for SFC $G_s$ \\
\cline{2-3}
& $\mathcal{C}_{SFC}$ & Set of all SFCs \\
\hline
\multicolumn{2}{c||}{VNF $v$} & Virtual Network Function $v$ \\
\hline
\multirow{2}{*}{VNF $v$} & $C_{v}$ & Capacity demand of VNF $v$ \\
\cline{2-3}
& $L$ & Locality matrix $V(G_S) \times V(G_I)$ \\
\hline
\multirow{6}{*}{$N_i$} & $\overline{C}_{N_i}$ & Total resource capacity of node $N_i$ \\
\cline{2-3}
& $p_{N_i}$ & Cost per resource unit used of node $N_i$ \\
\cline{2-3}
& $D_{AP, S}(N_i,N_j)$ & Delay between $N_i$ and $N_j \in V_{AP}(G_I) \cup V_{S}(G_I)$ \\
\cline{2-3}
& $D_{M_q}(N_i,N_j)$ & Delay between $N_i$ and $N_j \in V_{RC_q}(G_I)$ \\
\cline{2-3}
%& \multirow{2}{*}{$\mathbb{P}_{bat} (N_i,0)$} & Probability of having battery for the whole \\
%& & time interval using $0$ resources \\
%\cline{2-3}
& \multirow{2}{*}{$\mathbb{P}_{bat} (N_i,C_{N_i})$} & Probability of having battery for the whole \\
& & time interval using $C_{N_i}$ resources \\
\hline
\multirow{5}{*}{$AP_k$} & $d_{AP_{k}}$ & Delay for the coverage area of $AP_k$ \\
\cline{2-3}
& \multirow{2}{*}{$\mathbb{P}_{AP^q_k}(t_u)$} & Probability of cluster $q$ to be in the coverage \\ 
& & area of $AP_k$ in time subinterval $t_u$ \\
\cline{2-3}
& $p_{AP_k}$ & Cost of usage of $AP_k$ \\
\cline{2-3}
& $\kappa_{q}$ & Coverage probability threshold for cluster $q$ \\
\hline
%\multicolumn{2}{|c|}{\multirow{2}{*}{$\mathbb{P}_{bat} (N_i,C_{N_i})$}} & Probability of having battery for the whole \\
%\multicolumn{2}{|c|}{} & time interval using $\bar{C}_{N_i}$ resources \\
%\hline
\hline
\multicolumn{3}{c}{\textbf{Variables}} \\ 
\hline
\hline
\multicolumn{2}{c||}{$d_{G_s}(t_u)$} & Delay of SFC $G_s$ in time $t_u$ \\
\hline
\multicolumn{2}{c||}{$d(N_i, N_j, t_u)$} & Delay between nodes $N_i$ and $N_j$ in time $t_u$ \\
\hline
\multicolumn{2}{c||}{$x(v, N_i)$} & Placement of VNF $v$ in node $N_i$ \\
\hline
\multicolumn{2}{c||}{$C_{N_i}$} & Resource usage in node $N_i$ \\
\hline
\multicolumn{2}{c||}{$AP^q_{k}(t_u)$} & Usage of $AP_k$ by cluster $q$ time $t_u$ \\
\hline
\multicolumn{2}{c||}{$\mu: V(G_S) \mapsto V_H(G_I)$} & VNF to host node mapping structure \\
\hline
\multicolumn{2}{c||}{$\alpha: \{t_u\} \times \{q\} \mapsto V_{AP}(G_I)$} & AP selection structure for all clusters \\
\end{tabular}
\end{adjustbox}
% \caption{Definition of variables and parameters}
% \label{tab:notation}
\end{table*}
The network infrastructure is represented by a graph $G_I$, where the nodes $V(G_I)$ contain NR and LTE RUs, generally referred to as Access Points (APs) $V_{AP}(G_I)$, server nodes (representing edge or cloud servers) $V_S(G_I)$, and mobile nodes $V_M(G_I)$.
Hence, the vertex set of the graph is built up as $V(G_I)=V_{AP}(G_I)\cup V_S(G_I)\cup V_M(G_I)$. 
Host nodes $N_i$ with computation capacities $\overline{C}_{N_i}$ are stored in $V_H(G_I)=V_S(G_I)\cup V_M(G_I)$, and their corresponding unitary price is represented by $p_{N_i}$.
As a realistic generalization to the mobile robotics use case, the concurrent management of multiple robot clusters is assumed.
The subsets of $V_M(G_I)$ define the clusters of robots $V_{RC_q}(G_I) \subseteq V_M(G_I), 1 \leq q \leq Q$, where $Q$ refers to the number of clusters. Moreover, graph edges $E(G_I)$ represent the connections between the infrastructure nodes, which are annotated by their transmission delays.
% Edges between the nodes of robot clusters and the non-mobile (static) part of the infrastructure are not represented in $G_I$, as these connections change over time.
% Nodes of distinct robot clusters may communicate through APs and static infrastructure nodes.
% == Suggested rephrasing #1
% Since  robots within clusters move, their connections change, and edges between them are not represented in $G_I$.
% And they communicate with other clusters through APs.
% == suggested rphrase #2
% Due to the mobile clusters' mobility, their connections to the static part of the infrastructure are not represented by edges in $G_I$, because they should be defined by the optimization system's AP selection.
% == suggested rphrase #2 (no APs)
Due to the mobile clusters' mobility, their connections to the static part of the infrastructure are not represented by edges in $G_I$.

The mobile nodes $V_M(G_I)$ are connected to access points $V_{AP}(G_I)$ in order to communicate with other nodes of the infrastructure. However, the nodes are moving and may change the AP along time or be in an area where several access points have coverage. Thus, %the possibility of being in the coverage area of the different access points is modeled as a probability parameter ($\mathbb{P}_{AP^q_k}(t_u)$) and the optimization model decides to which AP a cluster of robots must be connected in each point of time guaranteeing a minimum of coverage all the time.
this work assumes that each AP has a coverage area $AP_k$, and 
the mobility pattern of robot cluster $q$ is modeled by the distribution of being in the AP coverage areas $\mathbb{P}_{AP^q_k}(t)$.
%the probability of a robot cluster $q$ to fall inside that coverage area is defined as $\mathbb{P}_{AP^q_k}(t)$,
The parameter $t$ is a time instant within an interval $(t_0,t_1)$ in which the network service will be running.
For the sake of simplicity in the model, the time interval is discretized in subintervals, thus continuous time $t$ becomes discrete time $t_u$. This discretization of time defines the handovers between access points during the whole interval, hence, in different subintervals of time the robot cluster's associated AP can be different.
The cost of using an AP for a single subinterval $t_u$ by any single cluster is $p_{AP_k}$.
%the AP to which is attached one cluster of robots can be different.
The energy consumption of the mobile nodes is modeled by the distribution $\mathbb{P}_{bat} (N_i,C_{N_i})$ depending on the allocated load to node $N_i$, which represents the probability of having a not depleted battery for the whole interval $(t_0,t_1)$.
% The battery is also modelled as the probability of having enough battery for the time interval considered to run the services depending on the resources consumed ($\mathbb{P}_{bat} (N_i,C_{N_i})$). It limits the amount of resources available for VNF usage.
% It is worth to mention that the formulation with probabilities allows to model in advance the posterior time interval, allowing to predict the deployment of services before the user starts the interaction with the system.
%\change{}{IF REALLY NEEDED, WE COULD REMOVE THIS SENTENCE AND THE TABLE TO SHORTEN THE PAPER-->}
Probability distributions $\mathbb{P}_{AP^q_k}(t)$ and $\mathbb{P}_{bat} (N_i,C_{N_i})$ describe the dynamics of the mobile part of the infrastructure, which are used by the optimization system for the service deployment in time interval $(t_0, t_1)$.

The requested Network Services are represented with a NS graph $G_S$, with the nodes being VNFs $v\in V(G_S)$ and their capacity requirements $C_v$.
Each Service Function Chain (SFC) is a subgraph $G_s \subseteq G_S$ with its own set of VNFs and path, as the one depicted in the NS graph of Figure~\ref{fig:use-case-illustration}, and expressed in Eq.~\eqref{eq:serviceGraph}.
%
% \scalebox{0.9}{
% \begin{minipage}
\begin{align}
    \mathcal{C}_{SFC} = \big\{(G_s, \Delta_{G_s}) ~ | ~ &V(G_s) \subseteq V(G_S), \label{eq:serviceGraph}\\
    &E(G_s) \in \mathcal{P}(G_S), \Delta_{G_s} \in \mathbb{R}^+\big\} \nonumber
\end{align}
% \end{minipage}
% }
%
where $\mathcal{C}_{SFC}$ represents the set of SFCs in Network Service $G_S$, and $\mathcal{P}(G_S)$ represents the paths of the NS graph $G_S$.
Each SFC has a corresponding delay requirement $\Delta_{G_s}$ which defines an upper bound of the total delay of the SFC path $E(G_s)$.

For a better understanding of the model, all the notation used for the mathematical formulation of the optimization problem are gathered in Table~\ref{tab:notation}.

\subsection{Optimization problem}
\label{sec:optimizationModel}
%\change{}{THIS SENTENCE CAN BE REMOVED TO SHORTEN THE PAPER-->}This section is devoted to mathematically model the requirements the system model demands in order to perform the deployment of services.
%\subsection{Computational resources}
%\label{subsec:resources}

\floatstyle{ruled}
\newfloat{form}{t!}{lop}
\floatname{form}{Formulation}

\begin{form}
    \begin{IEEEeqnarray}{ll}
        x(v, N_i) \in \{0, 1\} \qquad \forall v\in V(G_S), \forall N_i\in V(G_I) \label{eq:binary-variables} \\
        \sum_{N_i \in V(G_I)} x(v, N_i) = 1 \quad \forall v\in V(G_S) \label{eq:map-to-one-place} \\
         C_{N_i} \leq \overline{C}_{N_i}, \quad \forall N_i\in V(G_I) \label{eq:resource-limit} \\
         C_{N_i}=\sum_{v\in V_S(G_I)} x(v,N_i)C_v, \quad \forall N_i\in V(G_I) \label{eq:alloc-resources} \\
         x(v, N_i) \leq L(v, N_i), \quad \forall v \in V(G_S), \forall N_i \in V(G_I) \label{eq:locality} \\
         \sum_{AP_k \in V_{AP}(G_I)} \hspace{-1em}AP^q_k(t_u) = 1, ~ \forall 1 \leq q \leq Q, \forall t_u \in (t_0,t_1) \label{eq:only-one-ap-several-clusters} \\
        \sum_{AP_k \in V_{AP}(G_I)}\hspace{-1em} AP^q_k(t_u) \cdot \mathbb{P}_{AP^q_k}(t_u) \ge \kappa_q ,
        \quad \substack{\forall 1 \leq q \leq Q \\ t_u \in (t_0,t_1)} \label{eq:coverage-constraint-several-clusters} \\
%    d(N_i, N_j, t_u) =\begin{cases}
%    \mathlarger{\sum}\limits_{AP_k \in V_{AP}(G_I)} AP^q_k(t_u)[D_{M_q}(N_i, N^{(r_q)}) + d_{AP_k} + D_{AP,S}(AP_k, N_j)], &  \text{if}~ N_i \in V_{RC_q}(G_I) \land N_j \in V_S(G_I); \\
%    \hfill d(N_j, N_i, t_u) &  \text{if}~ N_j \in V_{RC_q}(G_I) \land N_i \in V_S(G_I); \\
%    \hfill D_{AP,S}(N_i, N_j) &  \text{if}~ N_i, N_j \in V_S(G_I) \cup V_{AP}(G_I); \\
%    \hfill D_{M_q}(N_i, N_j) &  \text{if}~ N_i, N_j \in V_{RC_q}(G_I); \\
%    \hfill \mathlarger{\sum}\limits_{AP_k \in V_{AP}(G_I)} \big( AP^{q_1}_k(t_u)[D_{M_{q_1}}(N_i, N^{(r_{q_1})}) + d_{AP_k}] + &  \text{if}~ N_i \in V_{RC_1}(G_I) \land\ N_j \in V_{RC_2}(G_I) \\
%    \hfill + AP^{q_2}_k(t_u)[D_{M_{q_2}}(N_i, N^{(r_{q_2})}) +d_{AP_k}] \big) + d_{AP_{k_1},AP_{k_2}}AP^{q_1}_{k_1}AP^{q_2}_{k_2},  & \\
%    \end{cases} \label{eq:delay} \\
    \rm{delay} \ \rm{equation}  \ \rm{is} \ \rm{presented} \ \rm{in} \ \rm{Eq.}~\eqref{eq:delay} \nonumber \\
    d_{G_s}(t_u)=\hspace{-1em} \sum_{\substack{(v_i,v_j) \in E(G_s) \\N_i, N_j \in V(G_I)}} \hspace{-1em}x(v_i,N_i) x(v_j,N_j) d(N_i, N_j, t_u) \label{eq:delayService} \\
    d_{G_s}(t_u) \leq \Delta_{G_s},~\forall (G_s, \Delta_{G_s}) \in \mathcal{C}_{SFC}, \forall t_u \in (t_0,t_1) \label{eq:delayThreshold} \\[1em]
    \mathbb{P}_{bat} (N_i,C_{N_i}) = \mathbb{P}_{bat} (N_i,0) -   \nonumber\\
    \quad- \dfrac{C_{N_i}}{\overline{C}_{N_i}} \left(\mathbb{P}_{bat} (N_i,0) -\mathbb{P}_{bat} (N_i,\overline{C}_{N_i})\right), \forall N_i \in V_M(G_I) \nonumber\\
    \label{eq:battery} \\
    \mathbb{P}_{bat} (N_i,C_{N_i}) \geq th_{bat}^{s},~ \forall N_i \in V_M(G_I),~\forall G_s \in \mathcal{C}_{SFC} \label{eq:batteryThreshold} \\[2em]
    \min \sum_{N_i \in V(G_I)} C_{N_i} \cdot p_{N_i} + \sum_{t_u,q,k} AP^q_k(t_u)\cdot p_{AP_k} \label{eq:objectiveFunction}
    \end{IEEEeqnarray}

    \caption{Optimization problem}
    \label{problem:optimizationModel}
\end{form}

\begin{figure*}[ht!]
    \begin{equation}
    \label{eq:delay}
    \resizebox{\textwidth}{!}{$%
d(N_i, N_j, t_u) =\begin{cases}
\hfill\mathlarger{\sum}\limits_{AP_k \in V_{AP}(G_I)} AP^q_k(t_u)[D_{M_q}(N_i, N^{(r_q)}) + d_{AP_k} + D_{AP,S}(AP_k, N_j)], &  \text{if}~ N_i \in V_{RC_q}(G_I) \land N_j \in V_S(G_I); \\%
\hfill d(N_j, N_i, t_u), &  \text{if}~ N_j \in V_{RC_q}(G_I) \land N_i \in V_S(G_I); \\%
\hfill D_{AP,S}(N_i, N_j), &  \text{if}~ N_i, N_j \in V_S(G_I) \cup V_{AP}(G_I); \\%
\hfill D_{M_q}(N_i, N_j), &  \text{if}~ N_i, N_j \in V_{RC_q}(G_I); \\%
\mathlarger{\sum}\limits_{\substack{AP_{k_1}, AP_{k_2}\\ \in V_{AP}(G_I)}} AP_{k_1}^{q_i}(t_u) AP_{k_2}^{q_j}(t_u) \left( D_{AP,S}(AP_{k_1},AP_{k_2}) + d_{AP_{k_1}} + d_{AP_{k_2}} + \mathlarger{\sum}\limits_{n\in\{i,j\}}  D_{M_{q_n}}(N_n,N^{(r_{qn})}) \right), & \text{if}~ N_i \in V_{RC_{q_i}}(G_I) \land\ N_j \in V_{RC_{q_j}}(G_I) \\%
% \hfill \mathlarger{\sum}\limits_{AP_{k_1}, AP_{k_2} \in V_{AP}(G_I)} \big( AP^{q_1}_{k_1}(t_u)[D_{M_{q_1}}(N_i, N^{(r_{q_1})}) + d_{AP_{k_1}}] + &  \text{if}~ N_i \in V_{RC_{q_1}}(G_I) \land\ N_j \in V_{RC_{q_2}}(G_I) \\
%     \hfill + AP^{q_2}_{k_2}(t_u)[D_{M_{q_2}}(N_i, N^{(r_{q_2})}) + d_{AP_{k_2}}] +  D_{AP,S}(AP_{k_1},AP_{k_2})AP^{q_1}_{k_1}(t_u)AP^{q_2}_{k_2}(t_u) \big)  & \\
% \big( AP^{q_1}_{k_1}(t_u)[D_{M_{q_1}}(N_i, N^{(r_{q_1})}) + d_{AP_{k_1}}] + &  \text{if}~ N_i \in V_{RC_{q_1}}(G_I) \land\ N_j \in V_{RC_{q_2}}(G_I) \\
%     \hfill + AP^{q_2}_{k_2}(t_u)[D_{M_{q_2}}(N_i, N^{(r_{q_2})}) + d_{AP_{k_2}}] +  D_{AP,S}(AP_{k_1},AP_{k_2})AP^{q_1}_{k_1}(t_u)AP^{q_2}_{k_2}(t_u) \big)  & \\
\end{cases} 
$%
}%
\end{equation}
% \vspace{0.3em}
% \hrule
\end{figure*}
%
%
%\boxcomment{BS: An intro like this would help to understand: Our optimization problem is formulated by Equations~\eqref{eq:resource-limit}-\eqref{eq:objectiveFunction} and described in details below.}
%
Our optimization problem is summarized in Formulation~\ref{problem:optimizationModel} and described in details below. 
The optimization must decide which infrastructure node $N_i \in V(G_I)$ should host which VNF $v\in V(G_S)$, this is represented by the binary decision variable $x(v,N_i)$ and constraints Eq.~\eqref{eq:binary-variables} and Eq.~\eqref{eq:map-to-one-place}.
%The optimization must decide whether to deploy a VNF $v\in V(G_S)$ at an infrastructure node $N_i \in V(G_I)$, which is represented by the binary decision variable $x(v,N_i)$.
%Such VNF $v$ asks for an amount of computational resources $C_v$ that will run inside $N_i$.
The resource capacities $\overline{C}_{N_i}$ must be respected by the load allocation on each node $N_i$.
%The system model has to prevent $N_i$ from running out of resources capacity $\overline{C}_{N_i}$. 
This requirement is gathered in Eq.~\eqref{eq:resource-limit},
%
%\begin{equation}
%     C_{N_i} \leq \overline{C}_{N_i},\quad \forall N_i\in V(G_I)
%    \label{eq:resource-limit}
%\end{equation}
where $C_{N_i}$ stands for the allocated resources in infrastructure node $N_i$ as presented in Eq.~\eqref{eq:alloc-resources}.
%\begin{equation}
%    C_{N_i} = \sum_{v\in V_S(G_I)} x(v,N_i)C_v, \quad \forall N_i\in V(G_I)
%    \label{eq:alloc-resources}
%\end{equation}

%To derive the deployment cost of NS $s$, each infrastructure node has assigned a unitary price $p_{N_i}$ for resource usage.

Furthermore, there may be a necessity of applying placement policies and VNF functional types. % and service access points with physical location. 
In order to include those policies in the model, the matrix $L(v, N_i)$ expresses locality constraints between the VNFs $v \in V(G_S)$ and infrastructure node $N_i \in V(G_I)$. 
Each element of the matrix is a binary constant, identifying whether the VNF can be located in an infrastructure node, as expressed in Eq.~\eqref{eq:locality}.
%
%\begin{equation}
%     x(v, N_i) \leq L(v, N_i), \quad \forall v \in V(G_S), \forall N_i \in V(G_I)
%    \label{eq:locality}
%\end{equation}

%\subsection{Coverage probability}
%\label{subsec:coverage}

%This work assumes that each AP has a coverage area $AP_k$, and the probability of a robot cluster to fall inside that coverage area is defined as $\mathbb{P}_{AP^q_k}(t)$, with $t$ being an instant within an interval $(t_0,t_1)$ in which the network service will be running.
%For the sake of simplicity in the model, the time interval is discretized in subintervals, thus continuous time $t$ becomes discrete time $t_u$. This discretization of time allows to define the changes of access point during the whole interval, hence, in different subintervals of time the AP to which is attached one cluster of robots can be different. 

\subsubsection{Radio coverage constraints}
\label{subsubsec:radio-constraints}
The deployment must also decide at each time interval to which access point each cluster of robots is attached to, that is, $AP^q_k(t_u)=1$ in case robot cluster $RC_q$ is connected to access point $AP_k$ at time $t_u$.
Eq.~\eqref{eq:only-one-ap-several-clusters} reflects the assumption that each cluster can only be attached to one AP at each interval. 
The deployment decision must also ensure that the coverage probability is above the imposed threshold $\kappa_q$ for mobile cluster $q$, as stated in  Eq.~\eqref{eq:coverage-constraint-several-clusters}.
%
%\begin{equation}
%    \sum_j AP^q_k(t_u) = 1, \quad \forall 1 \leq q \leq Q, \forall t_u \in (t_0,t_1)
%    \label{eq:only-one-ap-several-clusters}
%\end{equation}

%In addition, the deployment decision must ensure that the coverage probability is above the parameter imposed by the customer, $\kappa_q$. This requirement is gathered in Eq.~\eqref{eq:coverage-constraint-several-clusters}.
%
%\begin{equation}
%    \sum_j AP^q_k(t_u) \cdot \mathbb{P}_{AP^q_k}(t_u) \ge \kappa_q ,\quad \forall 1 \leq q \leq Q, t_u \in (t_0,t_1)
%    \label{eq:coverage-constraint-several-clusters}
%\end{equation}
%
%\subsection{Delay}
%\label{subsec:delay}
\subsubsection{Delay constraints}
\label{subsubsec:delay-constraints}
In order to measure the distances between infrastructure nodes, the metric used is the delay, which in the case of the static nodes is given in a matrix containing the precomputed and the time-independent delays, $D_{AP,S}(N_i, N_i) ~\forall N_i, N_j \in V_S(G_I) \cup V_{AP}(G_I)$. 
Similarly, the distances inside each mobile cluster are time invariant, precalculated and stored in matrix $D_{M_q}(M_i, M_j) ~\forall M_i, M_j \in V_{RC_q}(G_I), 1 \leq q \leq Q$.

Each access point $AP_k \in V_{AP}(G_I)$ provides a time- and distance-independent delay to its whole coverage area, its value is denoted by $d_{AP_k}$, while delay between APs is given with the value,\linebreak $D_{AP,S}(AP_{k_1},AP_{k_2})$.
The delay value between a mobile cluster and the static part of the infrastructure
and between mobile nodes belonging to two different clusters 
might vary according to the assigned APs during the time interval $(t_0, t_1)$.
%, as shown in Eq.~\eqref{eq:delay}. 
A mobile cluster $q$ has an appointed relay node $N^{(r_q)}$, which is connected to the APs, and all the traffic of other mobile nodes of the same cluster towards the fixed part of the infrastructure goes through the corresponding relay mobile node. Thus, the orchestration system can execute the handover of the cluster by only connecting the relay node to a different AP.
This way the delay of device-to-device communication is accounted in a different variable than the AP delays.
%(which is important, because the estimation errors of both cases depend on different parameters, and scale differently). 
Hence, the general delay function which covers any pair of infrastructure node types is expressed in Eq.~\eqref{eq:delay}.

The overall delay of a SFC $G_s \in \mathcal{C}_{SFC}$ in time $t_u$ is formulated in Eq.~\eqref{eq:delayService}, where the delays between the hosts of each SFC edge are summed.
The upper bound of the SFCs' total permitted delay~$\Delta_{G_s}$ for the whole optimization interval is expressed in constraint Eq.~\eqref{eq:delayThreshold}.
%
%\begin{equation} \label{eq:delayService}
%d_{G_s}(t_u) = \sum_{\substack{(v_i,v_j) \in \\ E(G_s)}}
%\sum_{\substack{N_i, N_j \in V(G_I)}}
%x(v_i,N_i) x(v_j,N_j) d(N_i, N_j, t_u)
%\end{equation}
% Furthermore, each SFC has an upper bound of their total permitted delay $\Delta_{G_s}$ for the whole optimization interval, which has to be considered in the model as the constraint formulated in Eq.~\eqref{eq:delayThreshold}.
%
%\begin{equation}
%    d_{G_s}(t_u) \leq \Delta_{G_s}, \quad \forall (G_s, \Delta_{G_s}) \in \mathcal{C}_{SFC}, \forall t_u \in (t_0,t_1)
%    \label{eq:delayThreshold}
%\end{equation}

%\subsection{Battery}
%\label{subsec:battery}

\subsubsection{Battery constraints}
\label{subsubsec:battery-constraints}
In order to place VNFs in mobile nodes it is necessary to ensure the mobile node will not run out of battery during the time interval $(t_0, t_1)$. Hence, the probability of having battery based on the resources used in the node is modeled in Eq.~\eqref{eq:battery},
%
%\begin{equation} \label{eq:battery}
%\begin{split}
%& \mathbb{P}_{bat} (N_i,C_{N_i}) = \mathbb{P}_{bat} (N_i,0) - \\
%\nonumber
%&- \dfrac{C_{N_i}}{\bar{C}_{N_i}} \left(\mathbb{P}_{bat} (N_i,0) -\mathbb{P}_{bat} (N_i,\bar{C}_{N_i})\right), \ \forall N_i \in V_M(G_I)
%\end{split}
%\end{equation}
where 
%$\overline{C}_{N_i}$ is the total capacity of the mobile node $N_i$ and
$C_{N_i}$ is the consumed capacity of mobile node $N_i$, and
$\mathbb{P}_{bat} (N_i,C_{N_i})$ is the probability of having not empty battery on $N_i$ by the end of time interval $(t_0, t_1)$ with $C_{N_i}$ allocated capacity.
The probability $\mathbb{P}_{bat} (N_i,C_{N_i})$ is a linear function between the empty and the fully loaded states.
To ensure the proper performance of the mobile nodes, the battery life is guaranteed in Eq.~\eqref{eq:batteryThreshold} by a threshold $th_{bat}^s$ given per SFC $G_s$, for all nodes hosting VNFs. 
This threshold takes into account the battery of all the mobile nodes hosting the VNFs of the service.
%
%\begin{equation} \label{eq:batteryThreshold}
%\mathbb{P}_{bat} (N_i,C_{N_i}) \geq th_{bat}^{s}, \quad \forall N_i \in V_M(G_I), \ \forall G_s \subseteq G_S
%\end{equation}

%\boxcomment{BN: replaced $\sqrt[|V(G_s)|]{th_{bat}^{s}}$ to align with the formulation in AMPL $th_{bat}^{s}$. }

%\boxcomment{BN: this should be only for those $N_i$-s which host any NF $v \in V_s(G_S)$}

%\subsection{Optimization problem}
%\label{subsec:optimization}

\subsubsection{Cost minimization}
\label{subsubsec:cost-minimization}
Finally, the problem minimizes the total cost of allocating the whole service $G_S$ demanded and AP usages by all of the mobile clusters. 
%In order to do this, it is necessary to introduce %$p_{N_i}$ for the price of a unit resource consumed in node $N_i$ and 
%the parameter $p_{AP_k}$ for the price of using each access point by a cluster for a single time subinterval $t_u$. 
Hence, the objective function is shown in Eq.~\eqref{eq:objectiveFunction}.
The VNF mapping $\mu$ and AP selection structures $\alpha$ are defined by the variables $x(v, N_i)$ and $AP_k^q(t_u)$ of a solution to the optimization problem.
This model is not linear in some equations as the one representing the delay in Eq.~\eqref{eq:delayService}, but each product of two variables can be easily linearized due to the fact that all the variables involved are binary variables. Thus, the linearization is performed by substituting each product of two binary variables by one extra binary variable.% as shown in Eq.~\eqref{eq:linearization}.
%
%\begin{equation} \label{eq:linearization}
%\begin{split}
%   & x(v_1,N_1), \ x(v_2,N_2), \ y(v_1,v_2,N_1,N_2) \in \{0,1\} \\
%   & y(v_1,v_2,N_1,N_2) = x(v_1,N_1) \cdot x(v_2,N_2), \\
%   & y(v_1,v_2,N_1,N_2) \leq x(v_1,N_1), \\
%   & y(v_1,v_2,N_1,N_2) \leq x(v_2,N_2), \\
%   & y(v_1,v_2,N_1,N_2) \geq x(v_1,N_1) + x(v_2,N_2) - 1
%\end{split}
%\end{equation}

%\boxcomment{Further assumptions:}
%\begin{itemize}
%    \item No bandwidth of the links will be considered in the formulation
%    \item reliability in case of wireless links is considered by their probability being in the reach of $AP_j$ by $\mathbb{P}_{AP_j}(t)$.
%    \item reliability of mobile nodes depend on the length of time and the allocated total amount of capacity on them. This is expressed by some probability
%    \item optimization objective is cost of deployment
%\end{itemize}

\section{Heuristic}
\label{sec:heuristics}

This section details the design of the heuristic which exploits the peculiarities of the system model to design an efficient and practical algorithm.

\subsection{Proposed heuristic}

The core idea of our heuristic algorithm is to use the fractional optimal solution of a bin packing problem of the VNFs and host nodes, which is deterministically rounded to an invalid integer solution. 
Next, the algorithm iteratively resolves the capacity, delay, battery and coverage constraint violations by changing the mapping location of VNFs in the initial invalid integer solution until a feasible mapping is found.

%\subsubsection{Bin packing with usage costs}

First, we introduce the bin packing problem variation with variable bin and item sizes supporting linear usage costs \cite{cambazard:hal-00858159} in Formulation~\ref{problem:binpacking}. Lemma~\ref{lemma:fractopt} states how to construct a fractional optimal solution for this bin packing variant, relaxing the integrality constraint. 
%Relaxing the integrality constraints on the VNF mapping variables $x(v, N_i) \in \{0, 1\}$ we get fractional mapping variables for all VNF-host node combinations $\Tilde{x}(v, N_i) \in [0.0, 1.0]$.
%\change{$C_{N_i}$ in the formulation}{to actual load or redefine here?}

\begin{problemformulation}[t!]
    \begin{align}
         \sum_{N_i \in V_H(G_I)} x(v, N_i) = 1   \quad&\forall v \in V(G_S) \\
         \sum_{v\in V(G_S)} x(v,N_i)C_{v} \le \overline{C}_{N_i} \quad& \forall N_i\in V_H(G_I)\\
         x(v, N_i) \in \{0,1\} \quad\forall v&\in V(G_S),  N_i\in V_H(G_I)\label{const:binpacking:integrality}\\
         min \sum_{N_i \in V_H(G_I)} C_{N_i} \cdot &p_{N_i}
    \end{align}
\caption{Bin Packing with Usage Cost \cite{cambazard:hal-00858159}\\
Input: VNFs $V(G_S)$ as items with weight, host nodes $V_H(G_I)$ as bins with capacity\\
Output: VNF placement respecting only capacity constraints}
\label{problem:binpacking}
\end{problemformulation}

%\pagebreak --- NB: why was it here?
\begin{lemma}[Fractional optimal solution of Formulation~\ref{problem:binpacking} \cite{cambazard:hal-00858159}]
    Let $\{a_i\}$ be a permutation of all host infrastructure nodes $N_i \in V_H(G_I)$ in ascending order by their unit costs of computation capacity $p_{a_1} \leq p_{a_2} \leq \dots \leq p_{a_{|V_H(G_I)|}}$. 
    Let $W_C = \sum_{v\in V(G_S)} C_{v}$ be the sum of all VNF capacities. 
    Let $b$ be the minimum number of host nodes in order $\{a_i\}$ where $\sum_{i=1}^b \overline{C}_{N_{a_i}} \geq W_C$. 
    % i.e. the sum of the first cheapest $b$ host node capacities is not lower than the sum of all VNF capacities.
    
    The fractional optimal solution (discarding the integrality constraint \eqref{const:binpacking:integrality}) of Formulation~\ref{problem:binpacking} is
    \begin{equation*}
        \Tilde{x}(v,N_{a_i}) = 
        \begin{cases}
            \frac{\overline{C}_{N_{a_i}}}{W_C} & \text{if} \quad i < b,\\
            \frac{W_C - \sum_{i=1}^{b-1}\overline{C}_{N_{a_i}}}{W_C} & \text{if} \quad i = b,\\
            0 & \text{if} \quad i > b;
        \end{cases}
        \quad \forall v \in V(G_S).
    \end{equation*}
    \label{lemma:fractopt}
\end{lemma}
\begin{proof}
    Refer to \cite{cambazard:hal-00858159}, where Formulation~\ref{problem:binpacking} is originally stated.
\end{proof}

%\subsubsection{Core VNF placement algorithm}

The proposed heuristic's core pseudo-code is shown in Algorithm~\ref{alg:vnfplacement}. 
Initially, the fractional optimal solution is retrieved and rounded to initial constraint-violating VNF placement, obeying only the locality constraints~\eqref{eq:locality} as shown in lines~\ref{line:fractoptrounding1}-\ref{line:fractoptrounding2}. 
The cost increasing order $\{a_i\}$ of mobile and server nodes are used from Lemma~\ref{lemma:fractopt} to involve additional hosts to the VNF placement pool, starting only from the first $b$ cheapest hosts.
In each iteration a set of violating items, respecting all 
%host capacity, SFC delay, robot cluster coverage and battery level
constraints is calculated based on the temporary decisions stored in the current VNF placement function $\mu$.
Next, the iteration in lines~\ref{line:innermost-iter1}-\ref{line:innermost-iter2} collects improvement scores for moving a VNF which is involved in any constraint violation to any currently considered host node (i.e. until index $b'$).
%The condition in line~\ref{line:meaningfulrelocation} filters the meaningful VNF relocations, while l
Line~\ref{line:heuristicfilter} heuristically filters only the VNF relocations whose improvement score is higher than a configured improvement score limit $\Upsilon$.
The improvement cost is calculated by the cost difference of VNF $v$ on the current host $\mu(v)$ and the possible new host $N_{a_i}$.
If any allowed VNF replacement is found, 
%the cheapest is chosen from set $\mathcal{I}$ and the current VNF mapping $\mu$ is updated, constraint violations are recalculated 
update actions are taken and a current AP selection $\alpha$ is retrieved as shown in lines~\ref{line:updatemapping1}-\ref{line:updatemapping2}.
Otherwise, the algorithm exits the improvement operations, and the next cheapest mobile or server node is included in the search by increasing $b'$.
If a feasible solution is found after any inner iteration (see line~\ref{line:solutionfound}), the procedure returns the current VNF placement function $\mu$ and AP selection structure $\alpha$.
The presented algorithm could be easily extended to continue searching for better quality solutions at the price of increased running time.

%\scalebox{0.9}{
%\begin{minipage}{\linewidth}
\begin{algorithm}[t!]
    \caption{\\
    Input: service graph $G_S$, infrastructure $G_I$, improvement score limit $\Upsilon$, and all constraints from Sec.~\ref{sec:problem-formulation}\\
    Output: VNF placement $\mu: V(G_S) \mapsto V_H(G_I)$ and AP selection $\alpha: \{t_u\} \times \{1\dots Q\} \mapsto V_{AP}(G_I)$ satisfying all constraints}
    \label{alg:vnfplacement}
    \begin{algorithmic}[1]
        \Procedure{PlaceVNFsSelectAPs}{$G_S, G_I, \Upsilon$}
        % NOTE: only use a_i in N index when we want to emphasize the order, otherwise use only i! 
        \State \parbox[t]{220pt}{$\Tilde{x}(v,N_i), b, \{a_i\} \gets$ fractional solution based on Lemma~\ref{lemma:fractopt} for host nodes $V_H(G_I)$ and VNFs $V(G_S)$ \strut} \label{line:fractoptrounding1}
        \For {$v \in V(G_S)$}     \Comment{Round initial solution}
            \State \parbox[t]{200pt}{$\mu(v) \gets argmax_{N_i \in V_H(G_I)} \Tilde{x}(v,N_i)$ which obeys locality constraints \eqref{eq:locality}\strut}
        \EndFor \label{line:fractoptrounding2}
        \For {$b' \in \{b \dots |V_H(G_I)|\}$} \Comment{In order of $\{a_i\}$} \label{line:core-dominating-part}
            \State $\mathcal{V}, \mathcal{R} \gets $ \textsc{violatingVNFMappins}($\mu$) %\Comment{Uses \textsc{chooseAPs}}
            \While {$\mathcal{V} \neq \emptyset$}
                \State $\mathcal{I} \gets \emptyset$ \Comment{Allowed improving VNF moves}
                \For{$v \in \mathcal{V}$, $i \in \{1 \dots b'\}$ \label{line:innermost-iter1}}
                    \State \parbox[t]{160pt}{\textbf{if} $\mu(v) \neq N_{a_i}$ \textbf{and} $\mu(v) = N_{a_i}$ obeys locality constraints \eqref{eq:locality} \label{line:meaningfulrelocation} \textbf{then}}
                        \If{$\Upsilon \leq$ \textsc{improveScore}($\mu$, $v$, $N_{a_i}$) \label{line:heuristicfilter}}
                            \State $impr\_cost \gets C_{v}(p_{N_{a_i}} - p_{\mu(v)})$
                            \State $\mathcal{I} \gets \mathcal{I} \cup \{(v, N_{a_i}, impr\_cost)\}$
                        \EndIf
                    \State \textbf{end if}
                \EndFor \label{line:innermost-iter2}
                \If{$\mathcal{I} \neq \emptyset$}
                    %\vspace{-10pt}
                    % \State \parbox[t]{185pt}{Let $\mu(v) = N_{a_i}$ where $(v, N_{a_i}, impr\_cost) \in \mathcal{I}$ is the cheapest improvement based on $impr\_cost$ \strut} \label{line:updatemapping1} 
                    \State \parbox[t]{180pt}{$\mu(v) \gets N_{a_i} ~|~ (v, N_{a_i}, impr\_cost) \in \mathcal{I}$ \textbf{and} $impr\_cost$ is minimal \strut} \label{line:updatemapping1} 
                    \State $\mathcal{V}, \mathcal{R} \gets $ \textsc{violatingVNFMappins}($\mu$)
                    \State \parbox[t]{180pt}{$\alpha \gets$ retrieve AP selection from violation record $\mathcal{R}$\strut} \label{line:updatemapping2}
                \Else \textbf{~break}
                \EndIf
            \EndWhile
            \If {AP selection $\alpha$ is valid \textbf{and} \\
            \qquad \qquad \qquad VNF placement $\mu$ is valid}
                \State \Return $\mu$, $\alpha$ \label{line:solutionfound} \Comment{Solution found}
            \EndIf
        \EndFor
        \State \Return $\emptyset, \emptyset$ \Comment{Solution not found}
        \EndProcedure
    \end{algorithmic}
\end{algorithm}
%\end{minipage}
%}

%\subsubsection{Calculation of constraint violations}
%Algorithm~\ref{alg:violatingVNFs} presents the details of constraint violation calculation used by the main procedure in Algorithm~\ref{alg:vnfplacement}.
All subsequently presented subroutines take the input 
%service and infrastructure graphs $G_S$ and $G_I$, respectively, together with all constraints and parameters as described in Sec.~\ref{sec:problem-formulation}, 
of Algorithm~\ref{alg:vnfplacement}, but these are omitted from the pseudo-codes for readability.
\textsc{violatingVNFMappings} takes as input the current VNF placement function $\mu$ and returns a set of violating VNFs $\mathcal{V}$ and an information storage of the actual constraint violations $\mathcal{R}$.
%All VNFs are added to set $\mathcal{V}$, which are currently placed to host nodes with violated capacity or battery constraints, as shown by the iterations of lines~\ref{line:capacityviolations} and \ref{line:batteryviolations}.
%The iterations starting at lines~\ref{line:capacityviolations} and \ref{line:batteryviolations} show how the VNF violation set $\mathcal{V}$ is expanded by the current capacity and battery constraint violations, respectively.
Based on the current VNF placement $\mu$, the feasibility of AP selection for each robot cluster $q \in \{1\dots Q\}$ is checked using the subroutine \textsc{chooseAPs}.
If the AP selection is not possible, all VNFs of the causing SFC $G_s$ are added to $\mathcal{V}$ and the violation information is stored in constraint violation record $\mathcal{R}$.

Algorithm~\ref{alg:chooseAPs} shows the details of how the AP selection and its feasibility based on the placement function $\mu$ are derived for a given robot cluster $q\in \{1\dots Q\}$ for all temporal subintervals.
Line~\ref{line:affectedSFCs} chooses the affected SFCs $G_s$, which have any VNF mapped to the mobile nodes of the robot cluster $q$.
Given the current VNF placement $\mu$, the total delay used by the path of the whole SFC $E(G_s)$ can be calculated using the delay expression \eqref{eq:delay}.
Access points are chosen by discarding the ones which do not meet the coverage requirement $\kappa_q$ and finding the one with minimal delay among the remaining ones: 
%\begin{equation}
%    AP_l = \mathrm{argmin}_{AP_k \in V_{AP}(G_I) \cap \{ AP_\phi : \mathbb{P}_{AP_\phi^q}(t_u) \geq \kappa_q\} } (d_{AP_k})
%    \label{eq:heuristic-ap-selection}
%\end{equation}
\begin{equation}
    AP_l = \mathrm{argmin}_{AP_k \in V_{AP}(G_I) \cap \{ AP_\phi : \mathbb{P}_{AP_\phi^q}(t_u) \geq \kappa_q\} } (d_{AP_k})
    \label{eq:heuristic-ap-selection}
\end{equation}
%\begin{equation}
%    AP_l = \{AP_{\phi} | d_{AP_{\phi}} \leq d_{AP_k}, \forall AP_{k}, AP_{\phi} \in V_{AP}(G_I) \wedge \mathbb{P}_{AP^q_k}(t_u) \geq \kappa_q \}
%    \label{eq:heuristic-ap-selection}
%\end{equation}
These operations are done by the function \textsc{delayDistWithCoverageAndAPSelection}, which also ensures that the same AP is chosen for a given input robot cluster $q\in \{1\dots Q\}$ in subinterval $t_u$, no matter which input SFC it gets.
The algorithm discards the impractical option of placing the VNFs of a single SFC to distinct mobile clusters. 
This simplification is only applied for the delay bounded VNFs, not to the other VNFs of the network service $G_S$.
If an access point $AP_l$ is found for subinterval $t_u$ with the given requirements, the selection is saved in AP selection function $\alpha$, otherwise the structure is invalidated and the reason is saved in $\mathcal{R}^{AP}$, as shown by the logical structure starting at line~\ref{line:APfoundornot}.
In case the computation capacities of a robot cluster are not used by any VNFs of any SFC, an access point still needs to be selected for the cluster, which is done by minimizing the cost instead of the unbounded delay and similarly filtering to the coverage probability (see line~\ref{line:mincostAPselection}).

\begin{algorithm}[t!]
    \caption{\\
    Input: Current VNF placemnet $\mu$, current (possibly incomplete or invalid) AP selection $\alpha$, robot cluster index $q$\\
    Output: Extended and/or ivalidated AP selection $\alpha$, AP selection violation record $\mathcal{R}^{AP}$}
    \label{alg:chooseAPs}
    \begin{algorithmic}[1]
        \Procedure{chooseAPs}{$\mu, \alpha, q$}
        \For{$t_u \in (t_0, t_1)$}
            \If{$\exists G_s \in \mathcal{C}_{SFC}, \exists v \in V(G_s)$, \\ 
            \qquad \qquad where $\mu(v) \in V_{RC_q}(G_I)$ \label{line:affectedSFCs}}
                \State \parbox[t]{180pt}{$d^{G_s}, AP_l \gets $ \textsc{delayDistWithCoverageAndAPSelection}($E(G_s), \mu, t_u, q, \kappa_q$)
                \strut}
%                 \State $d^{G_s} \gets 0$ \Comment{Sum shortest delay of $E(G_s)$ with $\mu$}
%                 \For{$(u, v) \in E(G_s))$}
%                     \State $d^{G_s},\quad AP \gets d^{G_s} +$\\ 
% \qquad \qquad \qquad \textsc{delayDistWithCoverage}($\mu(u), \mu(v), \kappa^q$)
%                 \EndFor
                \If{$d^{G_s} \leq \Delta_{G_s}$ \textbf{and} $\exists AP_l$ \label{line:APfoundornot}}
                    \State Let $\alpha(t_u, q) = AP_l$ \Comment{Same AP for all SFCs}%, because we minimize for the delay!}
                \Else 
                    \State Let $\alpha(t_u, q) = \lceil $
                    \State Add result $d^{G_s}$ and SFC $G_s$ to $\mathcal{R}^{AP}$
                    % \Return $\alpha, \mathcal{R}^{AP}$
                \EndIf
            \Else
                \State \parbox[t]{185pt}{Let $\alpha(t_u, q) = AP_l$, where $AP_l \in V_{AP}(G_I)$ and obeys coverage constraint~\eqref{eq:coverage-constraint-several-clusters} and $p_{AP_l}$ is minimal \label{line:mincostAPselection} \strut}
            \hspace{50pt}
            \EndIf
        \EndFor
        \State \Return $\alpha$, $\mathcal{R}^{AP}$
        \EndProcedure
    \end{algorithmic}
\end{algorithm}

%\subsubsection{Improvement score calculation}

Finally, the improvement score calculation is shown in Algorithm~\ref{alg:improvementScore}, which takes the current VNF placement $\mu$ and a possible relocation of VNF $v$ to $N_{a_i}$ as input, and outputs an integer whose higher value represents a more significant improvement.
The \textsc{improveScore} procedure uses the previously presented \textsc{violatingVNFMappings} function to evaluate how the mapping would change by the VNF mapping modification.
The mapping structure $\mu$ with less violating constraints is considered better, as shown in lines modifying the improvement score $y$. 
%In all constraint types the \textit{less} violating structure $\mu$ is considered better.
In case of capacity constraints, total improvement score $y$ would decrease, keep unchanged or increase if the number of hosts with more than their max capacity allocated would increase, stay or decrease by the VNF movement, respectively (see line~\ref{line:batteryImprovementScore}).
A similar score modification is done for each SFC, using the change in the number of temporal subintervals $t_u$ where the coverage or delay constraints are violated as shown by the iteration starting at line~\ref{line:delay-improvement-score}.
In the case of the battery constraints, the number of VNFs mapped to mobile nodes with violated battery thresholds are used.

\begin{algorithm}[t!]
    \caption{\\
    Input: Current VNF placement $\mu$, movement of VNF $v$ to host $N_{a_i}$\\
    Output: Integer in interval $[-|\mathcal{C}_{SFC}|-2, |\mathcal{C}_{SFC}|+2]$, the improvement score of the VNF movement}
    \label{alg:improvementScore}
    \begin{algorithmic}[1]
        \Procedure{improveScore}{$\mu$, $v$, $N_{a_i}$}
        \State $y \gets 0$ \Comment{Init. improvement score of moving $v$ to $N_{a_i}$}
        \State $\mathcal{V}, \mathcal{R} \gets $\textsc{violatingVNFMappings}($\mu$) %\Comment{Uses \textsc{chooseAPs}}
        \State $\mathcal{V}', \mathcal{R}' \gets $\textsc{violatingVNFMappings}($\mu~|~ \mu(v)=N_{a_i}$)
        \State \parbox[t]{220pt}{$y \gets y -1/+0/+1$ if \textit{number of hosts} $N_i$ with violated constraint \eqref{eq:resource-limit} increases/stays/decreases in $\mathcal{R}'$ compared to $\mathcal{R}$} \label{line:batteryImprovementScore}
        \For{$G_s \in \mathcal{C}_{SFC}$} \label{line:delay-improvement-score}
            \State \parbox[t]{195pt}{$y \gets y -1/+0/+1$ if \textit{number of subintervals} $t_u$ with any invalid mappings (i.e. where $\exists t_u, q : \alpha(t_u, q) = \lceil$~) increases/stays/decreases in $\mathcal{R}'$ compared to $\mathcal{R}$}
        \EndFor
        \State \parbox[t]{217pt}{$y \gets y -1/+0/+1$ if \textit{number of VNFs} $v$ which are mapped to any mobile node $V_M(G_I)$ with violated battery constraint~\eqref{eq:batteryThreshold} increases/stays/decreases in $\mathcal{V}'$ compared to $\mathcal{V}$}
        \State \Return $y$
        \EndProcedure
    \end{algorithmic}
\end{algorithm}

% \begin{itemize}
%     \item The delay to reach a server is  the max/min/avg of all incoming path delays
%     \item split e2e delay between request graph links (multiple approaches possible)
%     \item \comment{NB: pruning based on matching divided delay req with server reach delay?}
%     \item prune battery based on max capacity which can be run the whole interval
%     \item prune access points based coverage probabilty for all of the time instances
%     \item choose the active APs for each time instance based on their delay worst/best 
%     \item perform adapted variable bag and item size generalized bin packing problem with bin-dependent item profits
%     \item OR: we could do the steps in the reverse order:
%         \begin{itemize}
%             \item perform bin packing (using any applicapble heuristic) on the relaxed original problem (where we discard everything the bin packing variant cannot consider)
%             \item build/choose a solution which meets other constraints (i.e. pruning solutions)
%         \end{itemize}
%     \item Run the whole algorithm with multiple pruning strategies(??)
% \end{itemize}

\subsection{Complexity analysis}

A brief analysis on the heuristic's complexity and its termination is presented in Theorem~\ref{theorem:complexity} and its corresponding proof.

\begin{theorem}[Complexity of heuristic]
\label{theorem:complexity}
The overall complexity of the heuristic with positive improvement score limit $\Upsilon > 0$ is:
\begin{align}
    \mathcal{O}\Big( |V(G_S)|^4|V(G_I)|^3|\mathcal{C_{SFC}}|QT \Big)
\end{align}
where $Q$ and $T$ are the number of clusters and the number of subintervals $t_u$ in the optimization time frame $(t_0, t_1)$, respectively.
\end{theorem}

\begin{proof}
    %\commentinline{TODO: Fix linebreaks of proof @BN}
    Looking at Algorithm~\ref{alg:vnfplacement}, the fractional solution construction and its rounding are dominated by the iteration starting at line~\ref{line:core-dominating-part}, which is executed at most $|V(G_I)|$ times.
    Assuming a positive improvement score limit $\Upsilon$, the violating VNFs set $\mathcal{V} = \mathcal{O}(|V(G_S)|)$ decreases at least by one element in each iteration of the \emph{while} cycle.
    At most every iteration runs  \textsc{violatingVNFMappings}.
    Filtering for the allowed VNF movements in line~\ref{line:meaningfulrelocation} is done at most $\mathcal{O}(|V(G_S)||V(G_I)|)$ times, and in worst case for each of them we execute a \textsc{improveScore} subroutine call.
    These observations make Algorithm~\ref{alg:vnfplacement}'s complexity to be
    %\begin{align*}
        $\mathcal{O}\Big(|V(G_S)||V(G_I)|\big\{|V(G_S)||V(G_I)|$\textsc{improvScore} +                       \textsc{violatingVNFMappings}$\big\}\Big)$.
    %\end{align*}
    The \textsc{violatingVNFMappings}'s complexity is dominated by $Q\mathcal{O}(\textsc{chooseAPs})$, because the other constraints can be checked in $\mathcal{O}(|V(G_I)||V(G_S)|)$ time.
    Access point filtering for sufficient coverage in a longest SFC can be done in $\mathcal{O}(|V(G_I)||V(G_S)|)$ time, which is done for all SFCs $|\mathcal{C}_{SFC}|$, all SFC edges $\mathcal{O}(|V(G_S)|)$ for all time subintervals $T$.
    Which gives $\mathcal{O}(\textsc{violatingVNFMappings}) = \mathcal{O}(QT|V(G_S)^2||V(G_I)||\mathcal{C}_{SFC}|)$.
    Similarly, \textsc{improvScore} is dominated by \textsc{violatingVNFMappings}'s complexity.
    Finally, a Floyd-Warshall algorithm is used to pre-calculate the all the delay matrices $D_{AP,S}$ and $D_{M_q}$ with complexity $\mathcal{O}(|V(G_I|^3)$, which is dominated by the previous operations.
    Substituting and ordering the $\mathcal{O}(\cdot)$ notations, the statement follows.
\end{proof}

\section{Evaluation and results}
\label{sec:results}

This section compares the performance of Sec.~\ref{sec:heuristics} heuristic, with the optimal solution of Sec.~\ref{sec:problem-formulation} formulation from various aspects. 
As integer programs are generally impractical due to the hardness of the problem, our heuristic is extensively evaluated to demonstrate its applicability.
The heuristic solutions are compared to the optimal solution obtained with Gurobi which finds a solution within a gap optimality of 3\%.
Such comparison is done for the mobile robotics use case of Sec.~\ref{sec:use-case}, where scalability is a critical issue due to the size of the infrastructure and service graphs.

\subsection{Experiment setup}

% Rather than using a couple of SFC as the one depicted in Figure~\ref{fig:use-case-illustration}, every experiment generates a set of SFCs represented by series-parallel\footnote{SFC of Figure~\ref{fig:use-case-illustration} is an example of a series-parallel graph.} graphs~\cite{seriesParallel} (as in~\cite{seriesParallelvnfs}).

% comment#5.1 MOBILE CLUSTER MISUNDERSTUNDING
% \boxcomment{CB: When we say below "The SFCpaths start and end in the mobile cluster," what do we mean by "mobile cluster"? because it can be understood as that everything is running on the robots (the mobile elements), but this is not the case.}

%\commentinline{NB: addresses reviewer comment to show topology figure for the experiments.}
The presented evaluation scenario scales up the mobile robotics use case of Valencia city haven. 
A realistic 5G network infrastructure topology is considered with multiple types of wireless access points, while the service graph instances are random graphs. 
Many parameters of the experiment setting are examined during the presented simulations, varying the size of the input, SFC delay requirements, coverage probabilities and battery thresholds.

In order to generalize the service graphs and gain confidence in our simulations, series-parallel graphs are used to generate the network service topology $G_S$.
% SFCs paths, each one with its associated delay requirement (as stated in Eq.~\ref{eq:serviceGraph}).
Figure~\ref{fig:use-case-illustration} (left side) shows an example of such graph. 
% An arbitrary order series-parallel graph is randomly generated
% by starting from a loop edge that is recursively modified by either
% (i) subdividing an edge with a new node (\emph{series}); or (ii) adding a
% \emph{parallel} edge to an already existing one.
This graph class covers the structure of many data streaming applications, such as map-reduce topologies, and have been used in other realistic, industrial case studies for fog application allocation \cite{seriesParallelvnfs}.
%In the generated series-parallel service graph topology, nodes represent
%VNFs and edges represent Virtual Links (VLs).
The delay restriction applies to the whole SFC, which is a loop, starting and ending in
the lightweight VNFs running on the robots.
%  --- PREVIOUS WAY OF EXPLAINING SFC delay, lead to misunderstanding
%  --- comment#5.1
% The SFC paths start and end in the mobile cluster, so the paths' delays
% apply to the delay-bound service that is consumed at the lightweight VNFs
% running on the robots.
Among all the VNFs of the SFC, some of them are forced to run % NB: MAYBE this better reflects the test setting  -- I have added another sentence which addresses the exact placement req.
on top of the mobile robots' hardware $V_M(G_I)$, 
% on the robot cluster nodes $V_M(G_I)$
and the rest can run on top of any
% device with computational capabilities, i.e., servers $V_S(G_I)$ or other robots $V_M(G_I)$.
server $V_S(G_I)$ or robot $V_M(G_I)$.
It is up to the heuristic and the optimization formulation, to decide where to deploy them.

Every experiment uses the 5G infrastructure characterization of~\cite{luca-5g-qos}~and~\cite{ituArch}, which considers Ultra Reliable Low Latency Communications. %demanded by the mobile robots use case.
Table~\ref{table:infra} shows every infrastructure element considered
in the experiments, and Figure~\ref{fig:infrastructure}
illustrates the interconnection of the network
infrastructure.
Each M1 switch is located in the access ring of the
network, and it gathers the traffic of up to x6 LTE or
NR RUs.
Access rings have x6 M1 switches and x1 M2 switch,
all of them interconnected in a ring fashion.
Every M2 switch belongs to x4 access rings, and it
steers the traffic up to the aggregation ring, where
it is connected in a ring fashion with another x5
M2 switches. Experiments consider that edge and cloud
servers are reachable using M1 and M2 switches,
respectively.

\begin{figure}
    \centering
    \includegraphics[width=\columnwidth]{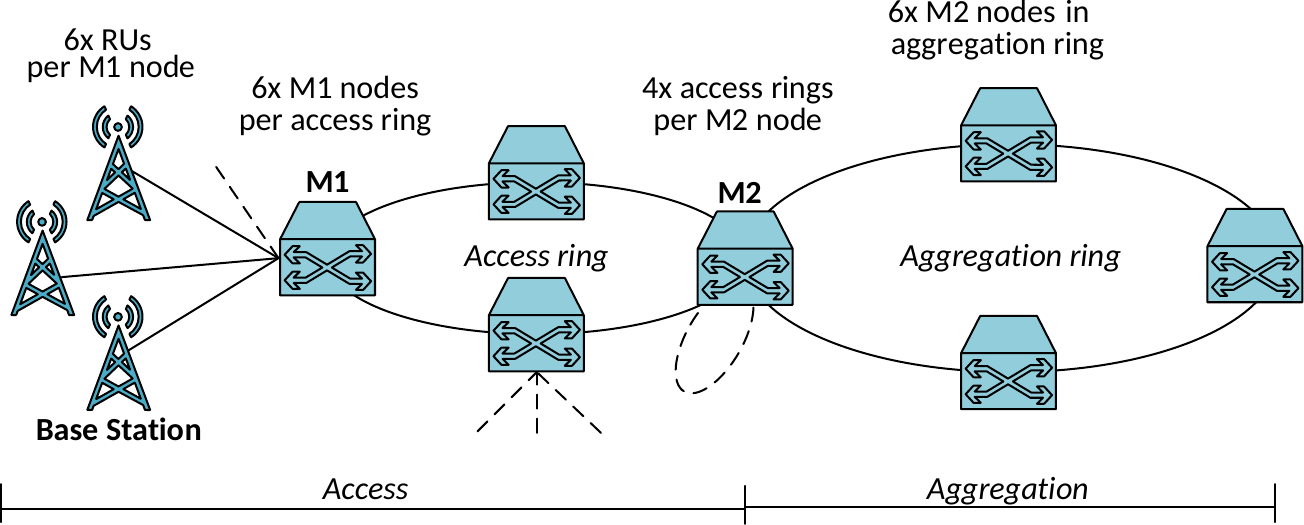}
    \caption{Interconnection of the experimentation network infrastructure.}
    \label{fig:infrastructure}
\end{figure}

Each point in the operation area of the robot cluster is covered at least by one LTE RU and at least by one NR RU.
The coverage probabilities for each time instance are derived by a function which maps the distance of the RU and the cluster to the coverage probability. 
The probability slightly decreases until the end of the RU coverage area,
and steeply drops to 0 at 120\% of the RU reach.
If a NR RU and the mobile cluster are not in LoS, the coverage probability
is 0, independent of their distance.
To achieve e2e delays demanded by the mobile robotics use case (between
1ms and 100ms), the experiment infrastructure assumes that
aggregation and access ring switches introduce packet processing delays between 1ms and 10ms, under the same characterisation as performed in~\cite{luca-5g-qos}.

%% == EXPLANATION ON M/M/1 DELAYS AND OVS cpus ==
%% To achieve e2e delays demanded by the mobile robotics use case (between
%% 1ms and 100ms), the experiment infrastructure assumes that
%% aggregation and access ring switches (denoted in~\cite{luca-5g-qos} as
%% M1/M2 switches, respectively) introduce packet processing delays between 1ms and 10ms.
%% This is possible assuming (i) an M/M/1 queueing discipline in M1/M2/switches; (ii) access/aggregation rings data rates reported in~\cite{luca-5g-qos} for the access and aggregation rings; (iii) Open vSwitch performance reported at~\cite{ovs}; (iv) 8vCPUs for M1 switches; and (v) 238 vCPUs for M2 switches.

To derive this section's results, a network infrastructure with just one cluster of robots has been generated with the 5GEN R package~\cite{5gen}.
%Latter, the infrastructure is stored as a GML~\cite{gml} file using the igraph~\cite{igraph} library.
Then, a Python script generates series-parallel NS graphs $G_S$ from which loop SFCs are selected. Robot cluster paths are encoded by coordinates which are used to calculate RUs coverage probabilities as robots move along the path.
Next, the Python script runs Sec.~\ref{sec:heuristics} heuristic to decide each VNF mapping on top of the infrastructure graph.
%, and invokes Gurobi 8.1 solver~\cite{gurobi} to obtain the optimal mapping.
Sec.~\ref{sec:optimizationModel} formulation is encoded in AMPL~\cite{ampl}, and the Python script invokes Gurobi 8.1 solver~\cite{gurobi} through the \emph{amplpy} API to obtain the optimal mapping.
All the experiments have been executed on two identical VMs with x4 vCPUs, 32GB of memory, and 132GB of disk.

\begin{table}
    \caption{Infrastructure used for experimentation}
    \label{table:infra}
    \begin{adjustbox}{max width=\columnwidth}
    \begin{tabular}[T]{r | l | p{55mm}}
        \textbf{\#} & \textbf{Element} & \textbf{Characteristics}\\ \hline \hline
        x2 & LTE RU~\cite{3gppLTE8} & 8km radio coverage, 5ms one way delay~\cite{ltertt}, 5.5 cost units OPEX~\cite{ltecost}\\ \hline
        x36 & NR RU~\cite{nrDelays} & 700m LoS coverage~\cite{nrCoverage}, 1ms one way delay, 11 cost units OPEX\\ \hline
        x10 & robots & x2 CPUs, 15.27 cost units/CPU~\cite{5gteconomical}\\ \hline
        x6 & edge server & x12 CPUs, 5.83 cost units/CPU~\cite{5gteconomical}\\ \hline
        x2 & cloud rack & 200 CPUs, 2.46 cost units/CPU~\cite{5gteconomical}\\ \hline
        x8 & M1 switch & x8 dedicated CPUs\\ \hline
        x6 & M2 switch & x238 dedicated CPUs\\ \hline
        x2 & access rings & fiber ring connection, $\le$6 M1 switches\\ \hline
        x1 & aggregation ring & fiber ring connection, x6 M2 switches
    \end{tabular}
    \end{adjustbox}
\end{table}
\raggedbottom

%% \begin{table}
%% \begin{threeparttable}[T]
%%     \caption{Infrastructure used for experimentation}
%%     \label{table:infra}
%%     \begin{tabular}{r | l | p{45mm}}
%%         \textbf{\#} & \textbf{Element} & \textbf{Characteristics}\\ \hline \hline
%%         x2 & LTE RU~\cite{3gppLTE8} & 8km radio coverage, 5ms one way delay~\cite{ltertt}, 5.5 cost units OPEX\tnote{3}\\ \hline
%%         x36 & NR RU\tnote{6} & 700m LoS coverage~\cite{nrCoverage}, 1ms one way delay, 11 cost units OPEX\\ \hline
%%         x10 & robots & x2 CPUs, 15.27 cost units/CPU\tnote{4}\\ \hline
%%         x6 & edge server & x12 CPUs, 5.83 cost units/CPU\tnote{4}\\ \hline
%%         x2 & cloud rack & 200 CPUs, 2.46 cost units/CPU\tnote{4}\\ \hline
%%         x8 & M1 switch & x8 dedicated CPUs\tnote{5}\\ \hline
%%         x6 & M2 switch & x238 dedicated CPUs\tnote{5}\\ \hline
%%         x2 & access rings & fiber ring connection, $\le$6 M1s\\ \hline
%%         x1 & aggregation ring & fiber ring connection, x6 M2s
%%     \end{tabular}
%%     \begin{tablenotes}\footnotesize
%%         \item[3] derived from~\cite{ltecost} OPEX study of LTE macro cells with 3 carriers
%%         \item[4] CPU costs are derived from~\cite{5gteconomical} techno-economical analysis
%%         \item[5] switch dedicated CPUs do not run VNFs
%%         \item[6] NR delays are based on numerologies 1 and 2 of~\cite{nrDelays}
%%     \end{tablenotes}
%% \end{threeparttable}
%% \end{table}
%% \raggedbottom

\subsection{Simulation results}
\begin{table*}[t]
    \centering
    \caption{Experiment parameters}
    \begin{adjustbox}{width=\textwidth}
    \begin{tabular}{c|c|c|c|c}
    \textbf{Parameter name/explanation}  & \multicolumn{4}{c}{\textbf{Value/range}} \\
    \hline
    \hline
    Experiment name & Scalability & Coverage & Delay & Battery \\ \hline
      Robot cluster path, see Figure~\ref{fig:use-case-illustration}   & S $\rightarrow$ D1 & S $\rightarrow$ D2 & S $\rightarrow$ D2 &  S $\rightarrow$ D1  \\ \hline
      Time interval count $t_u\in(t_0,t_1)$ & 24 & 24 & 24 & 24 \\ \hline
      Unloaded battery probability $\mathbb{P}_{bat} (N_i,0), \forall N_i\in V_M(G_I)$ & 99\% & 99\% & 99\% & 99\% \\ \hline
      Full loaded battery probability $\mathbb{P}_{bat} (N_i,\overline{C}_{N_i}), \forall N_i\in V_M(G_I)$ & 50\% & 80\% & 80\% & 50\% \\ \hline
      Battery probability threshold $th_{bat}^s$ & 40\% & 70\% & 70\% & 72\% \& 75\% \\ \hline
    %   Series-parallel graph generation ratio & \multicolumn{4}{c}{70\% series, 30\% parallel} \\ \hline
      Infrastructure delay sample count & 1 & 4 & 4 & 1 \\ \hline
      SFC delay [ms] $\Delta_{G_s}$ & 1000 & 5 & Varies & 1000 \\ \hline
      Randomized VNF vCPU requirement $C_v \forall v\in V(G_S)$ & 0.5 x \{0, \dots 4\} & 0.5 x \{0, \dots 4\} & 0.5 x \{0, \dots 4\} &  0.25 x \{0, \dots 4\} \\ \hline
      VNF count $|V(G_S)|$ & Varies & 10 & 10 & 26 \\ \hline
      VNF count bound to robots & 6 & 4 & 1 & Varies \\ \hline
      Coverage probability threshold $\kappa_q$ & 94\% & Varies & 94\% & 70\% \\ \hline
      Scenario repetition with different randomization seed & 14 & 24 & 20 & 14 \\ 
    \end{tabular}
    \end{adjustbox}
    \label{tab:simulation-params}
\end{table*}

This section presents the results of the extensive simulations, which studies the scalability and parameter sensibility of the previously presented algorithms. 
The details of the simulation parameters are shown in Table~\ref{tab:simulation-params} for each of the experiments. 
The cluster paths in the Valencia haven are represented by their source and target locations. 

All evaluation figures present boxplots, where the middle line shows the median (a.k.a. second quartile) of the dataset, while the body of the boxplots show the first and third quartiles (a.k.a. the medians of the first half and the second half of the dataset separated by its median).
The whiskers of the boxplots represent the datum which deviates from the boxplot body at most by 1.5 times the inter-quartile range, while outliers are individually plotted by circles which fall beyond the whiskers.

An input VNF placement problem with all previously presented constraints
is deemed feasible, if the AMPL implementation finds a valid solution that respects all
constraints in 30 minutes
(measured in wall-clock time).
In case of the heuristic, the timeout is reduced to 20 minutes.
All experiments were executed with $3\%$ optimality gap for AMPL, and with various improvement score limit values for the heuristic.

\newcommand{\colWidthMultiplier}{0.70}
\begin{figure*}
    \centering
    \subfloat[Scalability test: costs\label{fig:scalability:cost}]
    {\includegraphics[width=\colWidthMultiplier\columnwidth]{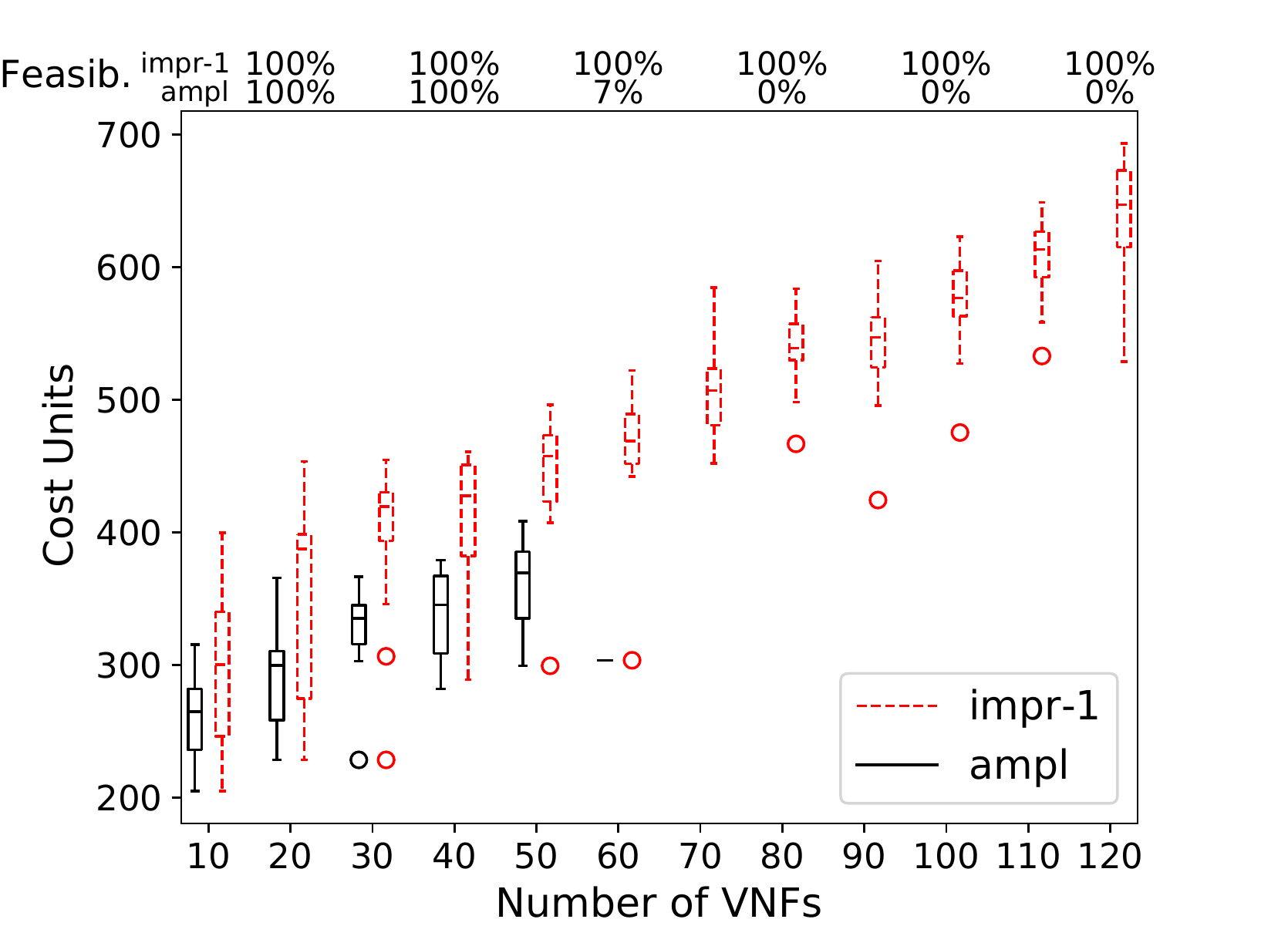}}
    \subfloat[Coverage threshold variation test: costs\label{fig:coverage:cost}]
    {\includegraphics[width=\colWidthMultiplier\columnwidth]{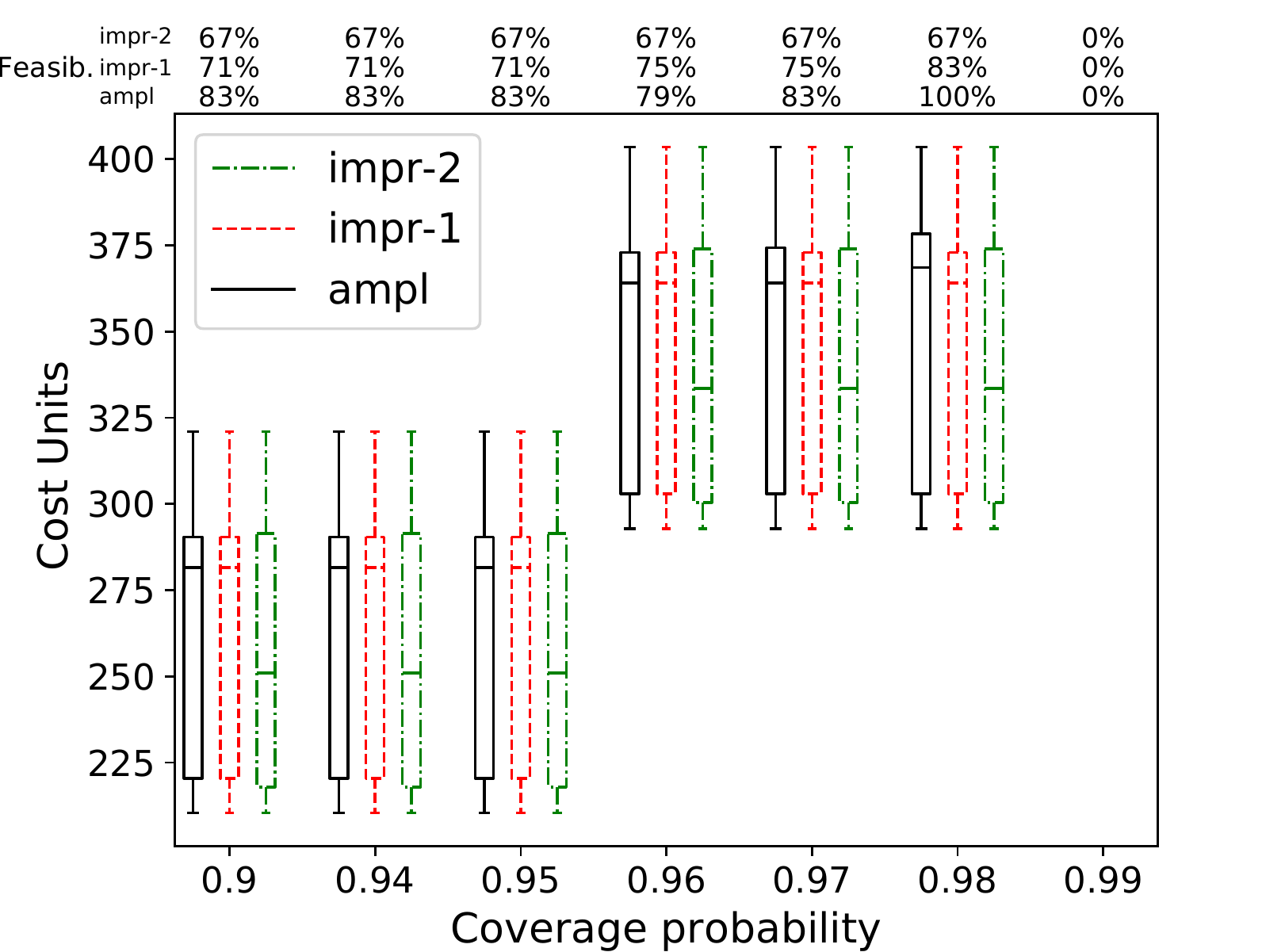}}
    \subfloat[SFC delay variation test: costs\label{fig:e2edelay:cost}]
    {\includegraphics[width=\colWidthMultiplier\columnwidth]{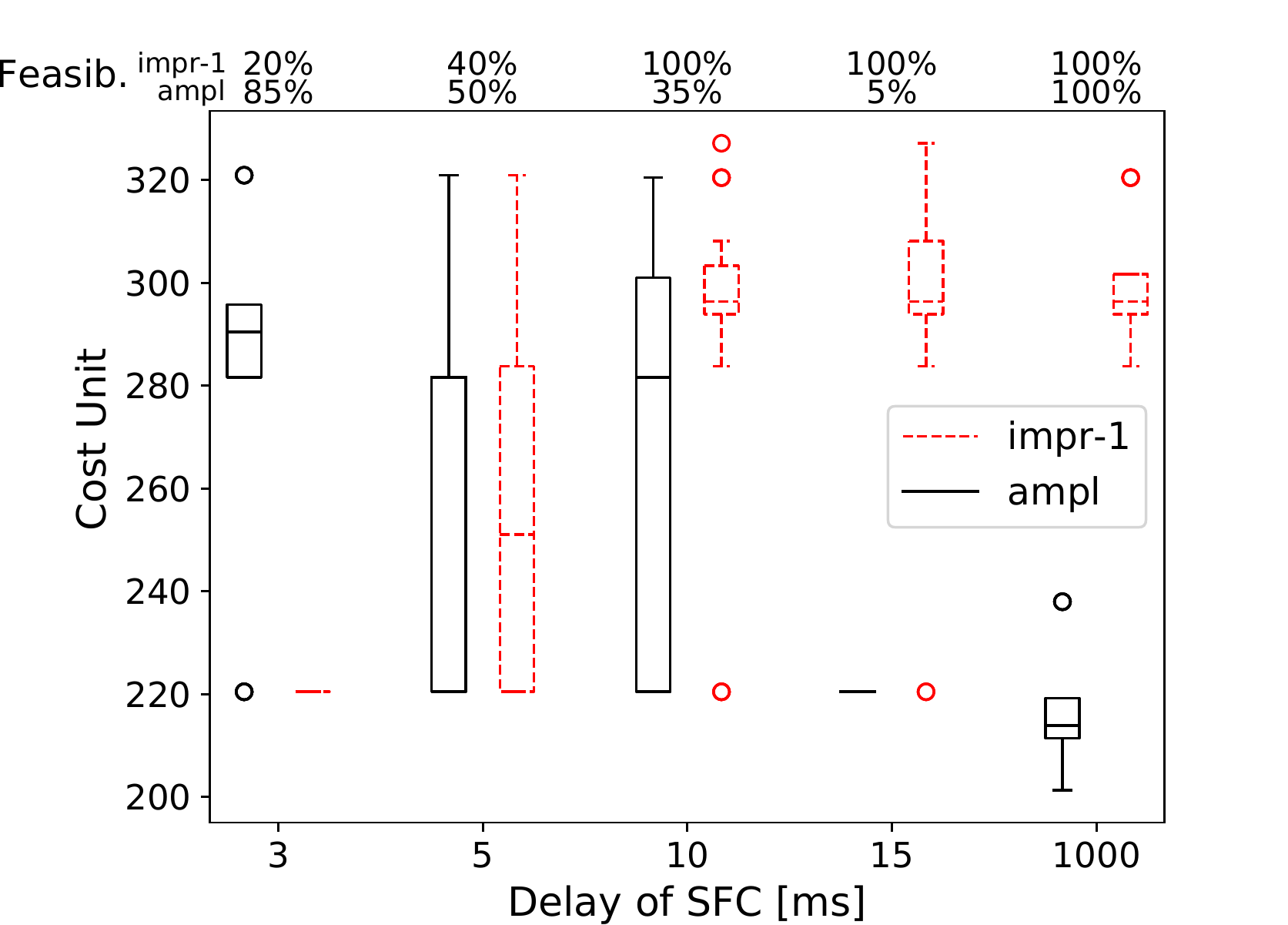}}
    \vspace{-15pt}
    
    \subfloat[Scalability test: runtimes\label{fig:scalability:time}]
    {\includegraphics[width=\colWidthMultiplier\columnwidth]{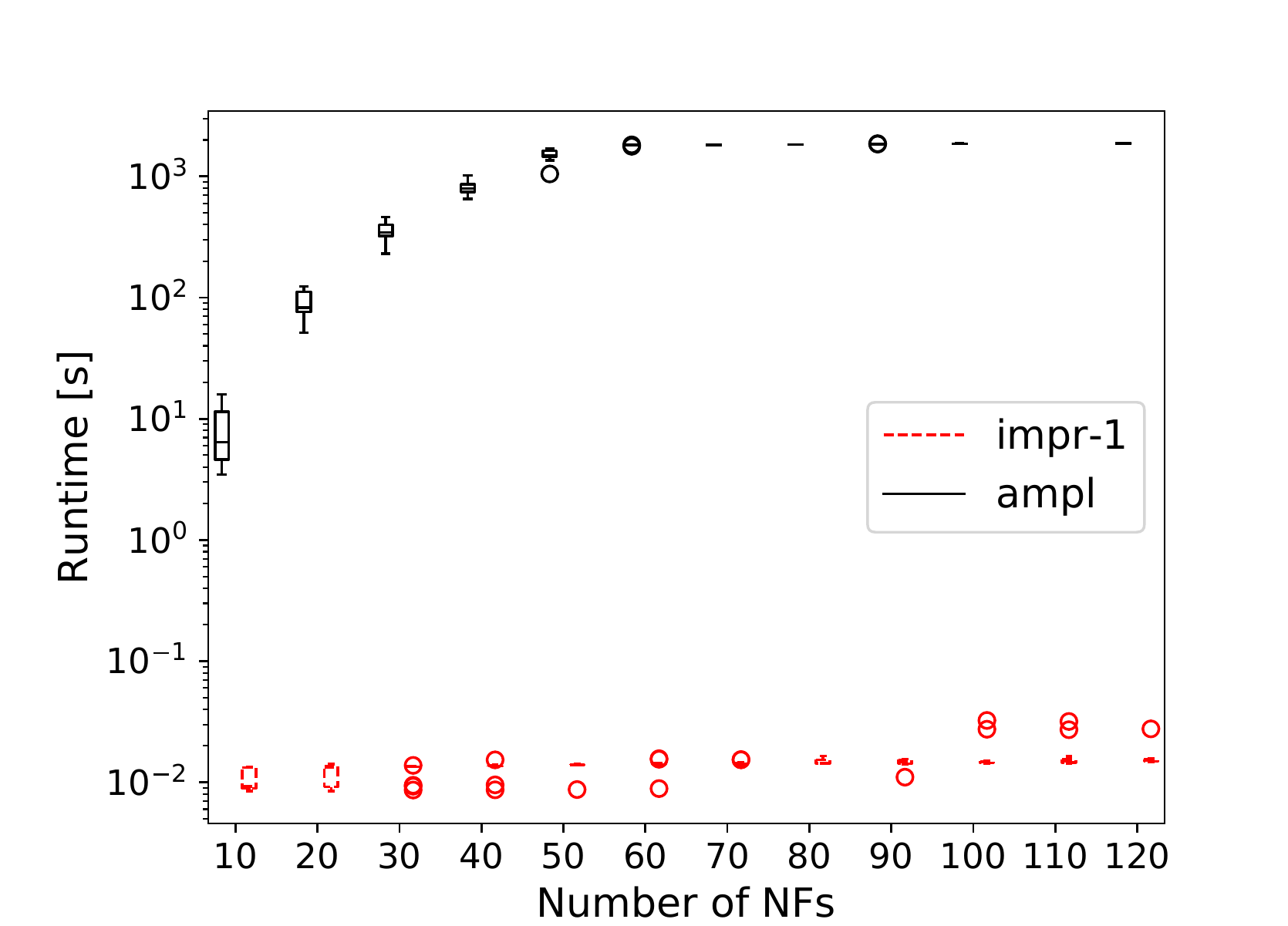}}
    %\subfloat[Running time of coverage threshold\newline variation test
    \subfloat[Coverage threshold variation test: runtimes\label{fig:coverage:time}]
    {\includegraphics[width=\colWidthMultiplier\columnwidth]{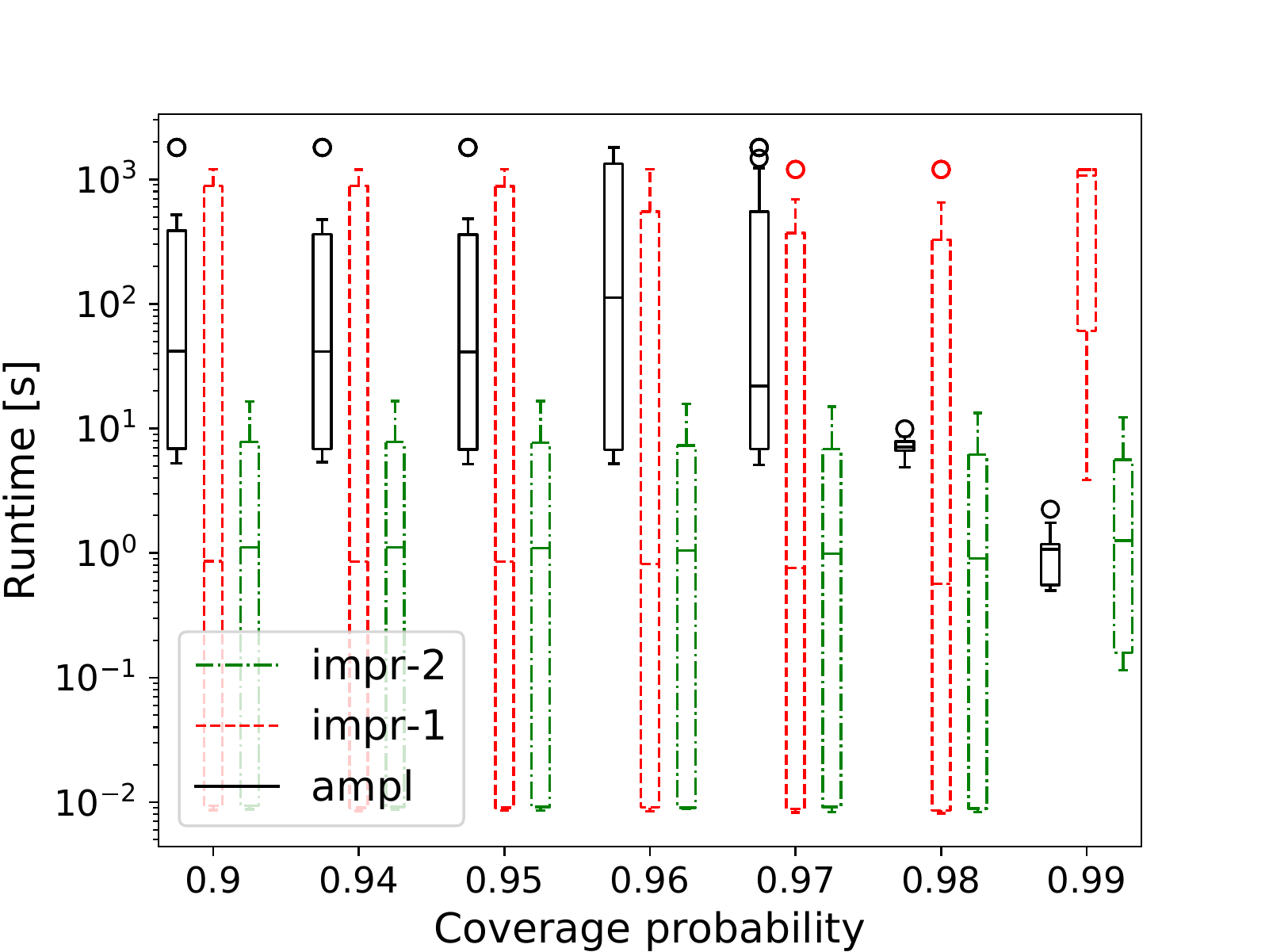}}
    \subfloat[SFC delay variation test: handover counts\label{fig:e2edelay:handover}]
    {\includegraphics[width=\colWidthMultiplier\columnwidth]{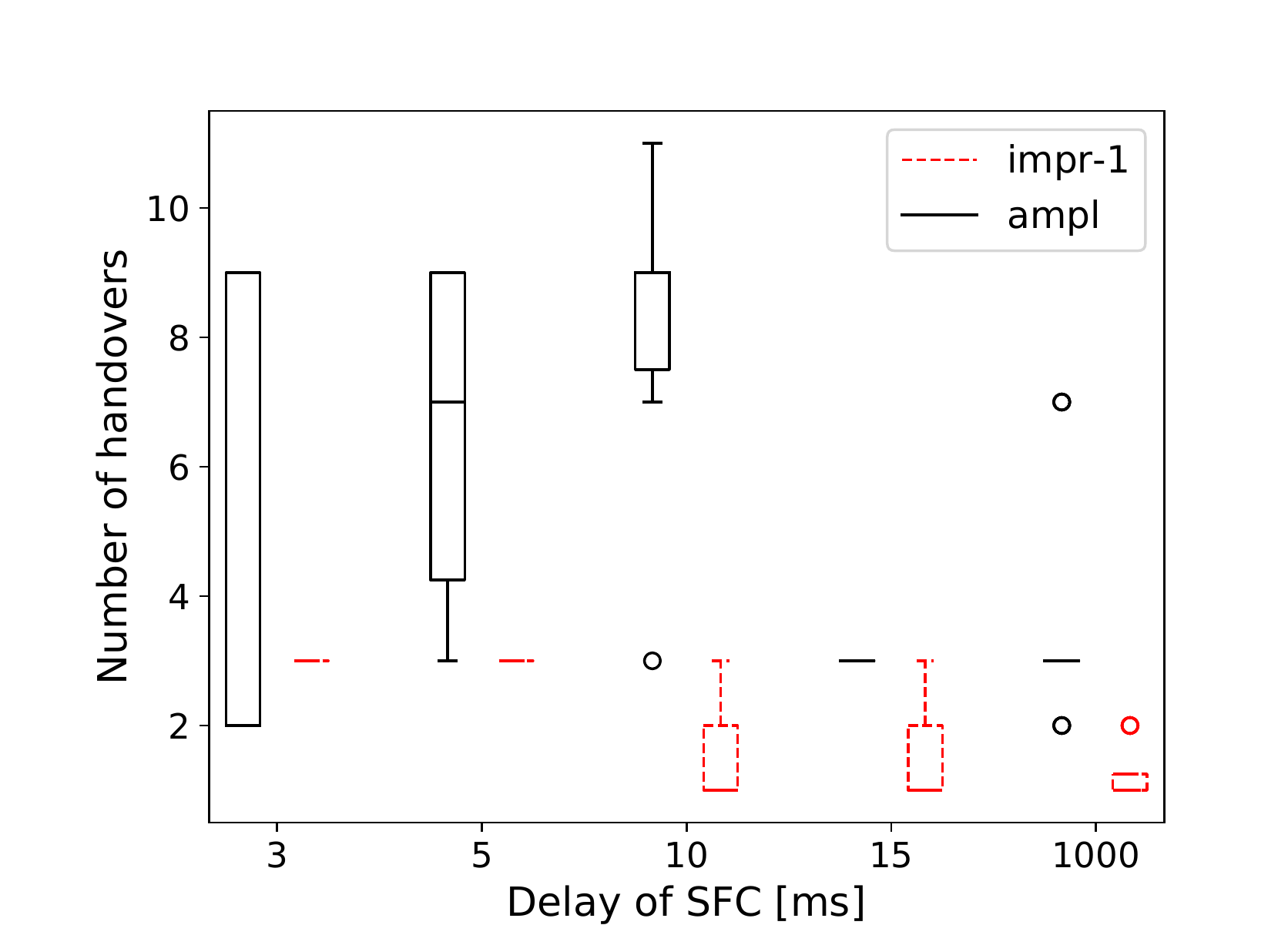}}
    \caption{Results of scalability, coverage probability and SFC delay experiments}
\end{figure*}

% Balazs's Version applying Antonio's comments for scalability section
First of all, the scalability of the algorithms are compared depending on the number of VNFs to be placed; results in terms of cost and runtime are shown in Figure~\ref{fig:scalability:cost} and Figure~\ref{fig:scalability:time}, respectively. 
The time-bound feasibility is shown on top of the figures for each randomized scenario repetition corresponding to the dependant value on the horizontal axis.
The scalability experiment is repeated multiple times for each input size, varying the distribution of VNF capacity requirements, the service graph's concrete topology, and the selection of the VNFs bound to the mobile cluster (see Table~\ref{tab:simulation-params}).
The scenario parameters allow a solution to be found in any randomized generation, due to loose SFC delay, coverage and battery probability thresholds; though the $30mins$ time limit may not be enough in all cases.
Figure~\ref{fig:scalability:cost} shows the time-bound feasibility ratio calculated on the randomized repetitions.
%for each number of VNFs. 
A steep drop of feasibility of the AMPL implementation
occurs at the VNF count of 60, which is due to reaching the computation timeout in each case.
The reason behind the timeouts is the exponential runtime of the
AMPL solution, which is shown by Figure~\ref{fig:scalability:time} in logarithmic time scale.
%{\color{red} REMOVE--Decreasing or increasing the VNF count results in 100\% and 0\% feasibility, respectively, due to timeout.
%Which shows exponential running time scaling of the AMPL implementation as depicted by the logarithmic scale of Figure~\ref{fig:scalability:time}.}
On the contrary, the heuristic with improvement score limit $\Upsilon$ set to $1$, shows an exceptionally good scaling both in terms of feasibility and running time, finding solutions below $100ms$ for all tests.
%The heuristic finds feasible solutions for all input sizes and shows linear runtime scaling for increasing number of VNFs as it is confirmed by data series \textit{impr-1} on Figure~\ref{fig:scalability:cost} and Figure~\ref{fig:scalability:time}.
Furthermore, Figure~\ref{fig:scalability:cost} depicts the quality of the heuristic solutions, which are between 15\% and 30\% of the optimal costs found by the integer program.

% Jorge's version fixing the fixed scalability section
% First, the scalability of the algorithms is evaluated based
% on the number of VNFs to be placed.
% Every experiment is repeated multiple times varying the number of VNFs,
% their capacity requirements, the SFC topology, and the number of VNFs bound to the mobile cluster.
% Figure~\ref{fig:scalability:cost} and Figure~\ref{fig:scalability:time}
% show results in terms of cost and running time, respectively. 
% And Figure~\ref{fig:scalability:cost} displays on top the percentage of
% experiments that found a feasible solution before the timeout.
% % write here cost related
% No matter the number of VNFs, Figure~\ref{fig:scalability:cost} shows
% that the heuristic with improvement score limit $\Upsilon=1$ (denoted
% as \emph{impr-1}) obtains deployment costs between 15\% and 30\% away
% from the optimal costs found by AMPL.
% Moreover, Figure~\ref{fig:scalability:time} illustrates that the heuristic
% implementation always takes less than 100ms to find a solution, whilst the
% runtime of the AMPL implementation increases exponentially as the number
% of VNFs increase.
% Indeed, for 80 VNFs running times are above 30 minutes for the AMPL
% implementation, which leads to a feasibility drop down to a 0\% (see
% Figure~\ref{fig:scalability:time}).

Second, the effect of the coverage probability threshold $\kappa_q$ is studied.
Figure~\ref{fig:coverage:cost} shows how the cost varies
by increasing the threshold, i.e. making the AP selection more strict.
As the coverage probability requirement increases, deployment costs become
more expensive, because the solutions impose the selection of the closer
and more expensive NR antennas, rather than the cheap LTE antennas.
Figure~\ref{fig:coverage:cost} depicts as well the feasibility, and shows
that for $\kappa_q=0.99$ all scenarios are infeasible, because there
exists at least one subinterval in which the cluster is not covered by
any antennas with such high probability.
Regarding the impact of the improvement score $\Upsilon$,
Figure~\ref{fig:coverage:cost} and Figure~\ref{fig:coverage:time} show 
that $\Upsilon=2$ (\emph{impr-2} time series) finds cheaper
solutions faster.
This is due to the heuristic design, which goes faster by shrinking the
solution space and considering only VNF relocations with higher improvement score.
The heuristic finds cheaper deployments faster, because they require less steps to make the rounded fractional solution feasible.
% The heuristic finds better deployment costs by rounding the fractional optimal
% solution to the cheapest (constraint violating) integer solution,
% so less heuristic steps are required to find a non-violating solution
% with lower cost.

%\boxcomment{Jorge: reviewed until here}

Next, the results of simulations varying the SFC delay are shown in Figure~\ref{fig:e2edelay:cost} and Figure~\ref{fig:e2edelay:handover}. 
The heuristic \textit{impr-1} struggles with finding feasible solutions in the allocated time for the $3ms$ scenarios, while AMPL manages to prove the existence of valid solutions as shown by the feasibility percentages of Figure~\ref{fig:e2edelay:cost}.
This could be easily addressed by introducing a search space pruning step in addition to the locality constraints. 
In the $3ms$ scenarios the usage of the cheap and high capacity cloud nodes is not an option because their RTTs from all APs are above this value.
Excluding these compute nodes from the allocation options for the VNFs contained in the strict SFCs would dramatically decrease the running time and thus increase the time-bound feasibility of the heuristic.
Although, additional pruning steps decrease solution quality in the cases of more permissive delay requirements.
On the other hand, the heuristic greatly outperforms the optimal solution search in the $10$-$15ms$ scenarios, where the AMPL algorithm fails to find any feasible solution.
In the loose SFC delay bound case of $1000ms$ both algorithms find solutions and their qualities are compared by their costs.
Another interesting aspect of the solutions are the number of required handovers needed for the whole optimization time interval. 
A lower handover count requires less management operations and results in a more stable service.
Handover comparison between the cost-optimal and the heuristic solutions are shown in Figure~\ref{fig:e2edelay:handover}.
The heuristic outperforms the optimal solution, which is especially relevant when the scenario could be solved by a few handovers as shown by the $10ms$ experiment scenarios with 100\% heuristic feasibility.
The AMPL algorithm could be modified to minimize the number of handovers, but it would further worsen its scalability, while the heuristic performs well by design.

\begin{figure}
    \centering
    \includegraphics[width=0.95\columnwidth]{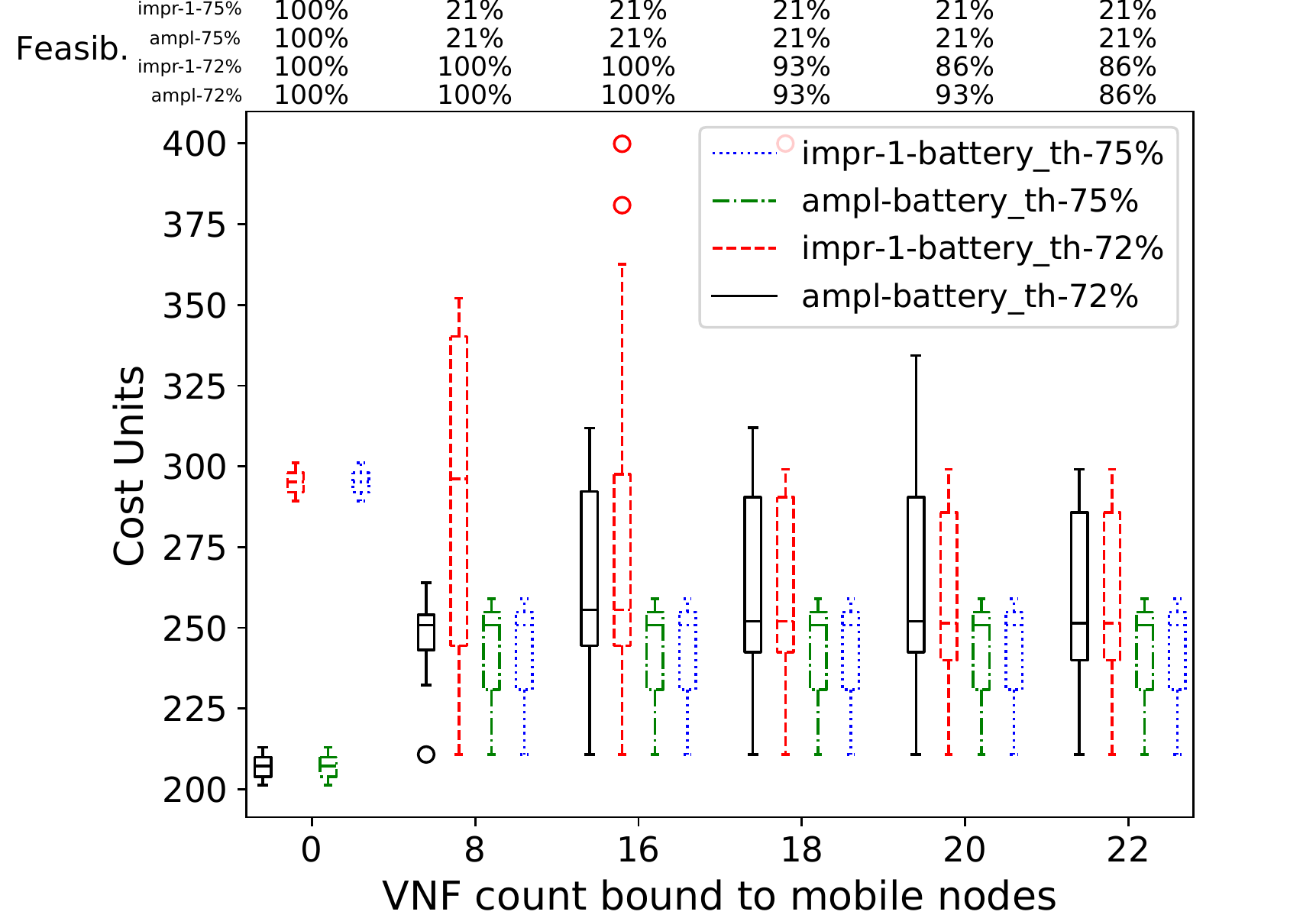}
    % \caption{The effect of varying the battery probability threshold on the cost and feasibility}
    \caption{Impact of battery probability threshold on cost and feasibility}
    \label{fig:battery:cost}
\end{figure}

Last, the results of the conducted experiments to examine the battery threshold parameter's effects are shown in Figure~\ref{fig:battery:cost}. 
The figure depicts cost values for both algorithms in cases of $72\%$ and $75\%$ battery alive probability requirements, as the number of VNFs to be placed on the mobile cluster increases.
These scenarios challenge constraint~\eqref{eq:batteryThreshold}, discovering the critical battery threshold to be around $72\%$-$75\%$.
In the $72\%$ case the scenarios are vastly feasible with a slight decrease.
The heuristic finds close to optimal allocations in almost all scenarios, except in the extreme case of much freedom.
In the more strict case of $75\%$, besides the no location-bound VNF experiment which is essentially the same as the $72\%$ case, the heuristic always finds all optimal solutions where it exists.
%\boxcomment{Why have we use 72 and 75 as values? I think we should explain it.}

The implementation of the algorithms, the simulation framework and all presented scenarios with raw data are available for further usage or result reproduction\footnote{GitHub link will be added to the camera ready version.}.

\section{Related work}
\label{sec:relatedwork}

Due to the widespread of virtualization technologies, the problem of allocating VNFs on top of physical resources has been of interest in recent years.
In most of the existing research the allocation of VNFs is envisioned as an optimization problem, that is generally $\mathcal{NP}$-hard \cite{rost2018VNEcomplexity}.

A common technique is to solve the VNF allocation problem as a variation
of the bin packing problem, taking the VNFs as items, and the bins as
servers.
% variable bin sizes, usage cost minimize energy we use,
% variable size bin-dependent costs, survey, don't consider delay -
% difficult to adapt 
Particularly, the first steps of this paper's proposed heuristic are build upon the basis defined in the algorithm of~\cite{cambazard:hal-00858159}, which minimizes a data center energy
consumption using a generalized bin packing problem.
Works as~\cite{vneResponseTime} solve the VNF allocation using the
variable size bin packing problem~\cite{variableSizeBinPacking}, which
provides an efficient solution to minimize both response time, and
resource utilization.
Other research have studied different and relevant generalizations for variable
sized bin-dependent costs \cite{Baldi2019GeneralizedBinP}.
% "even" should not be there, there are a tons of surveys, we want to refer to one, which is recent, and focuses on generalizations.
% In the literature there is even a survey~\cite{Coffman2013BinPA} posing a
% detailed categorization of bin packing problem variations, to solve the
% VNF allocation problem.
A recent survey categorizes bin packing problem generalizations which might be relevant to VNF placement solutions~\cite{Coffman2013BinPA}.
In general, algorithms for bin packing problems do not consider delays on the sequence of items, 
%(i.e. VNFs) 
nor any topological constraint among the bins, 
%(i.e. compute nodes)
so using their results for VNF placement problems is not trivial; our heuristic builds on such results.

Solutions of the VNF allocation problem must reshape with the new
5G networks, which bring computational capabilities closer to the user
thanks to MEC~\cite{mecIn5g}, and fog computing~\cite{fogOrchMagazine}.
Indeed, servers are way closer to antennas, or even co-located with them in
the edge, and IoT devices are becoming part of a dense network.
Thus, 5G comes with the urge of a more dense radio coverage, and the
possibility of sharing public/private network infrastructure~\cite{whitepaperPrivate}~\cite{3gpp5gReq} can help to achieve it.
Orchestration in the edge of 5G pops up the quest of deploying NSs with
very strict latency requirements, and research
as~\cite{eurecomGenetic5gt} and~\cite{latencyStoppingTheory} study
solutions about how to allocate VNFs in the edge.
While the first one uses a genetic algorithm to obtain a fixed allocation that minimizes/maximizes latency/availability, the second provides a stopping theory solution that migrates the allocated VNFs as time passes, such that latency restrictions are not violated.

There is research that focus on VNF allocation in fog
environments.
\cite{fogLatencyModel} presents an allocation model accounting for the computational overhead of fog devices, based on the assigned workload; and
\cite{trafficGenerators}~studies how to satisfy e2e delay by reducing the distance of the deployed service, to the users consuming the service (envisioned as traffic generators).
Another approach to allocate VNFs is to deploy them jointly using cloud and fog devices, as~\cite{jointRadioResources} does.
In that work, service providers derive a wireless and resource sharing model of fog devices, and the allocation is done using a student project allocation algorithm.
There are other results~\cite{Samie2016} related to low energy IoT devices, that study the trade-off between the energy requirement for computation, and transmitting data, as a computation task outsourcing pipeline is proposed.

% skipped recite: ~\cite{latencyStoppingTheory}
Although the literature already provides solutions to perform the VNF
allocation on edge and fog scenarios, this paper contributes to the state-of-the-art by tackling all at once
the (i) radio coverage; (ii) battery consumption; and (iii) e2e
delay restrictions present in 5G use cases with mobile compute nodes.

%%%%%%%%%% ORIGINAL BALAZS IDEA
% \begin{itemize}
%     \item mention orchestration algorithms;
%     \item public/private networks;
%     \item look for other fog-alike solutions;
%     \item reference working standardization items on public/private networks;
%     \item Refer to some binpacking problems, which cannot handle many of our constraints, but as our solution uses a bin packing variant we need to mention some \cite{genbinpacking-Baldi-2017}, actually used: \cite{cambazard:hal-00858159}. Add a ref to bin packing survey which studies variable sizes and/or costs.
% \end{itemize}

\section{Conclusions}
\label{sec:conclusion}

This paper has analyzed the notoriously hard problem of VNF placement in a realistic use case based scenario: mobile robotics for warehousing solution in the Valencia city haven, where a NPN deployment in a public network is assumed. 
In this scenario, mobile compute nodes act as an extension of the cloud and edge computing infrastructure, which triggers the need for VNF placement solutions with strict delay bounds and reliability constraints, while taking into account radio coverage, mobility and battery conditions.

The paper has introduced a system model and a mathematical formulation of the problem, to then propose an efficient heuristic building on the fractional optimal solution of a bin packing variant.
The heuristic has been extensively evaluated via simulations in terms of scalability and the strictness of constraints which are relevant to the use case. 
Results show that the proposed heuristic provides results with costs close to the optimal solution, while improving the convergence speed to the solution (therefore the number of time-feasible solutions is increased) and minimizing the number of required handovers.
To the best of our knowledge, our solution is the first to tackle the VNF placement problem simultaneously respecting battery, coverage and delay constraints over a mobile and volatile 5G infrastructure.

%\begin{itemize}
%    \item mention our conclusions on the research;
%    \item BW cost introduction + BW constraints;
%    \item access point (AP) cost;
%    \item individuality of mobile nodes;
%    \item 
%    \item battery model update;
%\end{itemize}

% \begin{IEEEkeywords}
% 5G, MEC, Point Process, Deployment, Characterization, Network Slicing,  Streaming, Low Latency, Augmented Reality, Virtual Reality.
% \end{IEEEkeywords}
% % \blfootnote{A footnote without marker}
% \input{introduction.tex}
% \input{related-work.tex}
% \input{model.tex}
% \input{scenario.tex}
% \input{results.tex}
% \input{conclusions.tex}
% \input{appendix.tex}
%\input{ideas.tex}

\section*{Acknowledgements}
This work has been partially funded by the \textit{EC H2020 5G-DIVE Project} (grant no.:~859881) and \textit{H2020-ICT-2018-3 EC H2020 5G-GROWTH Project} (grant no.:~H2020-ICT-2018-3 856709).
\bibliographystyle{IEEEtran}
\bibliography{references.bib}

%%%%%%%%%%%%%%%%
% BIBLIOGRAPHY %
%%%%%%%%%%%%%%%%
\newcommand{\mystring}{Lorem ipsum dolor sit amet, consectetuer adipiscing elit. Ut purus elit,
vestibulum ut, placerat ac, adipiscing vitae, felis. Curabitur dictum gravida
mauris. Nam arcu libero, nonummy eget, consectetuer id, vulputate a, magna.
Donec vehicula augue eu neque. Pellentesque habitant morbi tristique senectus
et netus et malesuada fames ac turpis egestas. Mauris ut leo. Cras viverra
metus rhoncus sem. Nulla et lectus vestibulum urna fringilla ultrices. Phasellus
eu tellus sit amet tortor gravida placerat. Integer sapien est, iaculis in, pretium
quis, viverra ac, nunc. Praesent eget sem vel leo ultrices bibendum. Aenean
faucibus. Morbi dolor nulla, malesuada eu, pulvinar at, mollis ac, nulla. 
Curabitur auctor semper nulla. Donec varius orci eget risus. Duis nibh mi, congue
eu, accumsan eleifend, sagittis quis, diam. Duis eget orci sit amet orci dignissim
rutrum.Lorem ipsum dolor sit amet, consectetuer adipiscing elit. Ut purus elit,
vestibulum ut, placerat ac, adipiscing vitae, felis. Curabitur dictum gravida
mauris. Nam arcu libero, nonummy eget, consectetuer id, vulputate a, magna.
Donec vehicula augue eu neque. Pellentesque habitant morbi tristique senectus
et netus et malesuada fames ac turpis egestas. Mauris ut leo. Cras viverra
metus rhoncus sem. Nulla et lectus vestibulum urna fringilla ultrices. Phasellus
eu tellus sit amet tortor gravida placerat. Integer sapien est, iaculis in, pretium
quis, viverra ac, nunc. Praesent eget sem vel leo ultrices bibendum. Aenean
faucibus. Morbi dolor nulla, malesuada eu, pulvinar at, mollis ac, nulla. 
Curabitur auctor semper nulla. Donec varius orci eget risus. Duis nibh mi, congue
eu, accumsan eleifend, sagittis quis, diam. Duis eget orci sit amet orci dignissim
rutrum.}

\newcommand{\biolimit}[0]{150}

% Balázs Németh
\begin{IEEEbiography}[{\includegraphics[width=1in,height=1.25in,clip,keepaspectratio]{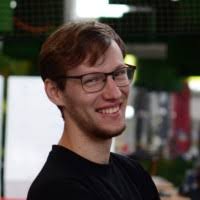}}]{Bal\'azs N\'emeth}
is an Industrial Ph.D. student at Budapest
  University of Technology and Economics (BME) in cooperation with
  Ericsson Research. He obtained his M.Sc. degree at BME as a Computer
  Science Engineer in info-communication specialization (2016). He has
  been working on orchestration algorithms for the H2020 5G-PPP 5G
  Exchange (5GEx) project. Currently, he is pursuing his Ph.D. degree in
  network softwarization with special focus on orchestration
  algorithms and next generation network models.
\end{IEEEbiography}

% Nuria Molner
\begin{IEEEbiography}[{\includegraphics[width=1in,height=1.25in,clip,keepaspectratio]{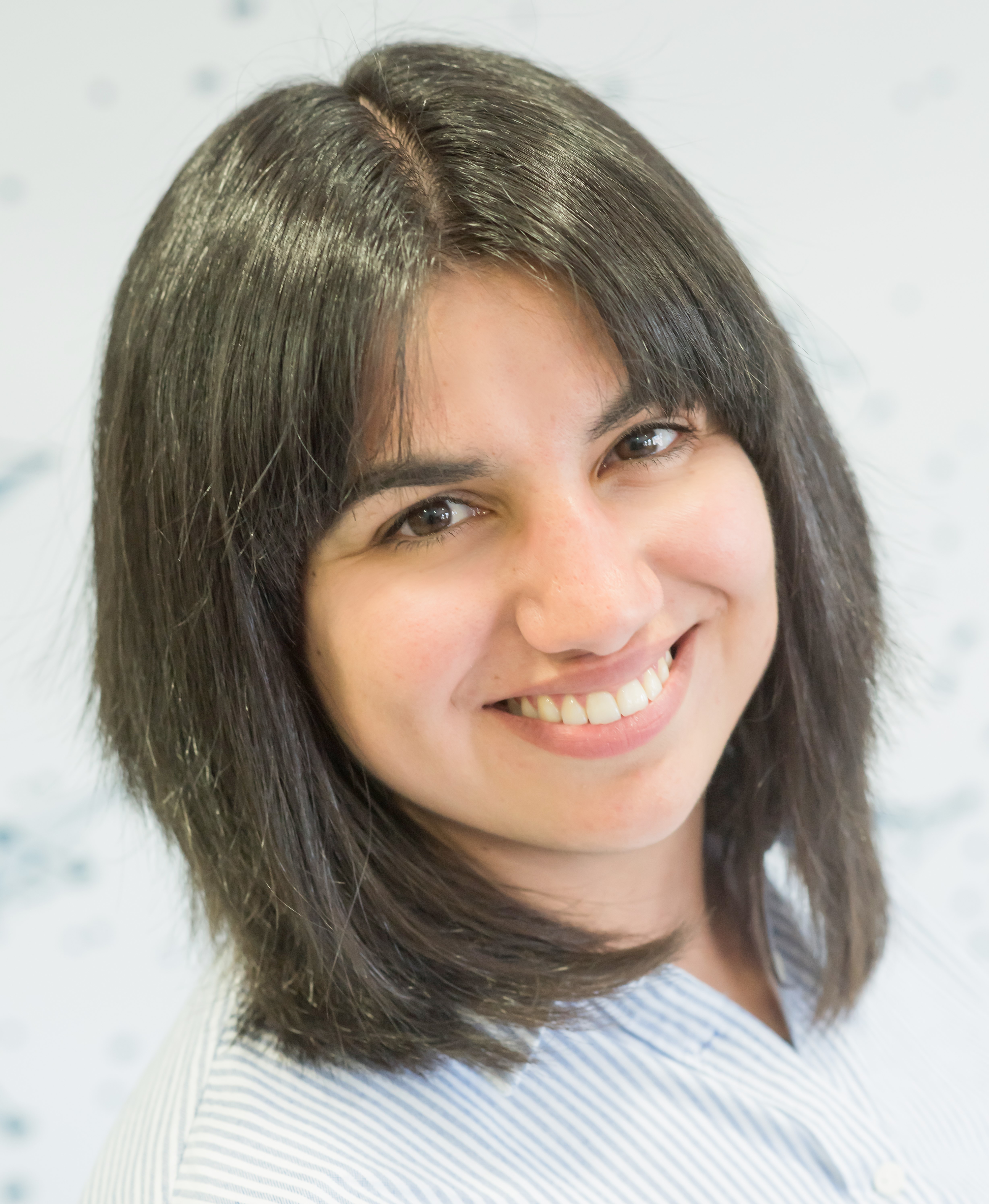}}]{Nuria Molner}
obtained her BSc in Mathematics from University of Valencia (Spain) in 2015 and her MSc in Telematics Engineering in University Carlos III of Madrid (Spain) in 2016. Currently, she is studying her PhD in Telematics Engineering from University Carlos III of Madrid, in the topic of optimization of 5G networks, while developing her research work at IMDEA Networks Institute and participating in European H2020 5G projects.
\end{IEEEbiography}

% Jorge
\begin{IEEEbiography}[{\includegraphics[width=1in,height=1.25in,clip,keepaspectratio]{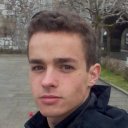}}]{Jorge Mart\'in P\'erez}
is a telematics engineering Ph.D. student at Universidad Carlos
III de Madrid. He obtained a BSc in mathematics in 2016, and
a BSc in computer science in 2016, both at Universidad
Autónoma de Madird. Latter on 2017 he obtained a M.Sc. in 
telematics engineering in Universidad Carlos III de Madrid.
His research focuses in optimal resource allocation in 5G
networks, and since 2016 he participates in
EU funded research projects.
\end{IEEEbiography}

\vspace{5em}

% Carlos J. Bernardos
\begin{IEEEbiography}[{\includegraphics[width=1in,height=1.25in,clip,keepaspectratio]{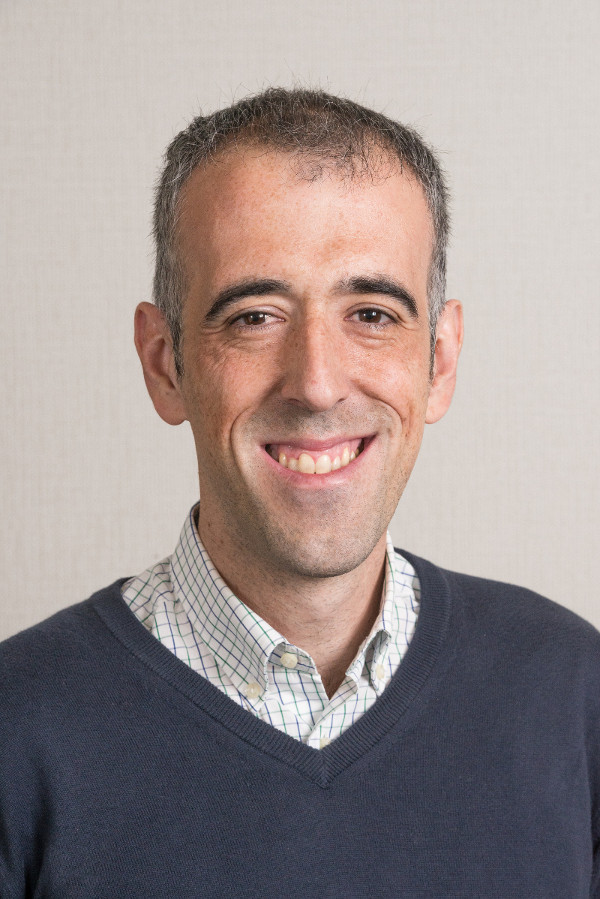}}]{Carlos J. Bernardos}
received a Telecommunication Engineering degree in 2003, and a PhD in Telematics in 2006, both from the University Carlos III of Madrid, where he worked as a research and teaching assistant from 2003 to 2008 and, since then, has worked as an Associate Professor. His research interests include IP mobility management, network virtualization, cloud computing, vehicular communications and experimental evaluation of mobile wireless networks. He has published over 70 scientific papers in international journals and conferences. He has participated in several EU funded projects, being the project coordinator of 5G-TRANSFORMER and 5Growth.
\end{IEEEbiography}

% Antonio de la Oliva
\begin{IEEEbiography}[{\includegraphics[width=1in,height=1.25in,clip,keepaspectratio]{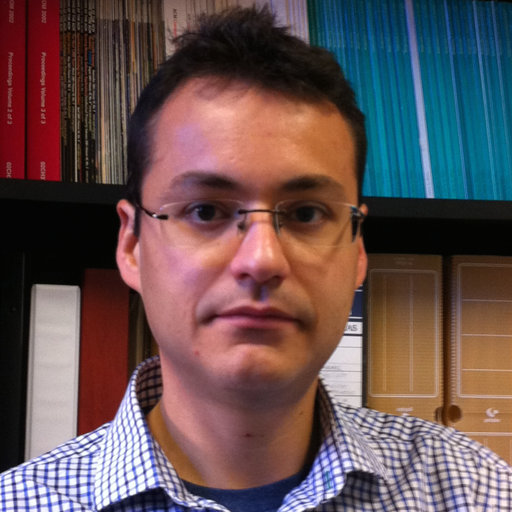}}]{Antonio de la Oliva}
received  his telecommunications  engineering  degree  in  2004 and his Ph.D. in 2008 from the Universidad Carlos III Madrid (UC3M), Spain, where he has been an associate professor since then. He is an active contributor to IEEE 802 wherehe has served as Vice-Chair of IEEE 802.21b and Technical  Editor  of  IEEE  802.21d.  He  has  also served as a Editor of several journals and magazines. He has published more than 50 papers on different networking areas. He has participated in several Eu funded project and he is currently coordinating the EC/TW joint project 5G-DIVE.
\end{IEEEbiography}

\vspace{-5em}
% Balázs Sonkoly
\begin{IEEEbiography}
  [{\includegraphics[width=1in,height=1.25in,clip,keepaspectratio]{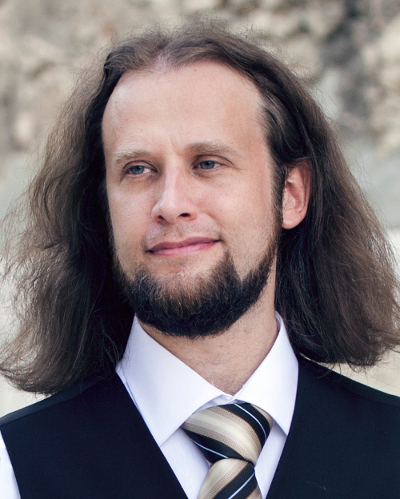}}]
  {Bal\'{a}zs Sonkoly} is an associate professor at Budapest
  University of Technology and Economics (BME) and he is the head of
  MTA-BME Network Softwarization Research Group.  He received his
  Ph.D.  (2010) and M.Sc. (2002) degrees in Computer Science from
  BME. He has participated in several EU projects (FP7 OpenLab, FP7
  UNIFY, H2020 5G Exchange) and national projects.  He was the demo
  co-chair of ACM SIGCOMM 2018, EWSDN'15,'14, IEEE HPSR'15. His
  current research activity focuses on cloud / edge / fog computing,
  NFV, SDN, and 5G.
\end{IEEEbiography}

\end{document}